\documentclass[11pt]{amsart}
\usepackage{amsmath}
\usepackage[dvips]{epsfig}
\theoremstyle{plain}
\newtheorem{theorem}{Theorem}[section]
\newtheorem{lemma}[theorem]{Lemma}

\newtheorem{cor}[theorem]{Corollary}

\theoremstyle{definition}
\newtheorem{remark}[theorem]{Remark}

\numberwithin{equation}{section}
\newcommand{\volg}{1+ u^{2}| \nabla f |_{g}^{2}} %volume form of \bg%
\newcommand{\volbarg}{1-u^{2}| \overline{\nabla} f |_{\overline{g}}^{2}}  %volume form of  g%
\newcommand{\barna}{\overline{\nabla}} %barred nabla - covariant derivative w.r.t. barred g%
\newcommand{\tilna}{\widetilde{\nabla}} %tilde nabla - covariant derivative w.r.t. tilde{g}-4metric%
\newcommand{\bg}{\overline{g}}  %barred g%
\newcommand{\tg}{\widetilde{g}} %tilde g - Lorentzian 4-metric%
\newcommand{\by}{\overline{Y}}  %barred Y%
\newcommand{\byp}{Y^{\phi}}  % Y^\phi%
\newcommand{\bk}{\overline{k}} %barred k - second fundamental form for t=0%
\newcommand{\be}{\overline{E}}
\newcommand{\bm}{\overline{B}}
\newcommand{\bF}{\overline{F}}

\begin{document}

\title[Deformations of Charged Axially Symmetric Initial Data] {Deformations of Charged Axially Symmetric Initial Data and the Mass-Angular Momentum-Charge Inequality}

\author[Cha]{Ye Sle Cha}
\address{Department of Mathematics\\
Stony Brook University\\
Stony Brook, NY 11794, USA}
\curraddr{Mathematical Institute\\
University of Oxford\\
Oxford, OX2 6GG, UK}
\email{ycha@math.sunysb.edu, cha@maths.ox.ac.uk}

\author[Khuri]{Marcus A. Khuri}
\address{Department of Mathematics\\
Stony Brook University\\
Stony Brook, NY 11794, USA}
\email{khuri@math.sunysb.edu}

\thanks{The second author acknowledges the support of
NSF Grants DMS-1007156 and DMS-1308753. This paper is also based upon work supported by NSF under
Grant No. 0932078 000, while the authors were in residence at
the Mathematical Sciences Research Institute in Berkeley, California, during
the fall of 2013.}

\begin{abstract}
We show how to reduce the general formulation of the mass-angular momentum-charge inequality, for axisymmetric initial data
of the Einstein-Maxwell equations, to the known maximal case whenever a geometrically motivated system of equations admits a solution. It is also shown that the same reduction argument applies to the basic inequality yielding a lower bound for the area of black holes in terms of mass, angular momentum, and charge. This extends previous work by the authors \cite{ChaKhuri}, in which the role of charge was omitted. Lastly, we improve upon the hypotheses required for the mass-angular momentum-charge inequality in the maximal case.

\end{abstract}
\maketitle

\section{Introduction}
\label{sec1} \setcounter{equation}{0}
\setcounter{section}{1}

Let $(M, g, k, E, B)$ be an initial data set for the Einstein-Maxwell equations. This consists of a 3-manifold $M$, Riemannian metric $g$, symmetric 2-tensor $k$ representing the
extrinsic curvature (second fundamental form) of the embedding into spacetime, and vector fields $E$ and $B$ representing the electric and magnetic fields, all of which satisfy the constraint equations
\begin{align}\label{1}
\begin{split}
16\pi\mu_{EM} &= R+(Tr_{g}k)^{2}-|k|_{g}^{2}-2(|E|_{g}^{2}+|B|_{g}^{2}),\\
8\pi J_{EM} &= div_{g}(k-(Tr_{g}k)g)+2E\times B,\\
div_{g} E & =div_{g} B=0.
\end{split}
\end{align}
Here $\mu_{EM}$ and $J_{EM}$ are the energy and momentum densities of the matter fields after the contributions from the electromagnetic field have been removed, $R$ is the scalar curvature of $g$, and $(E\times B)_{i}=\epsilon_{ijl}E^{j}B^{l}$ is the cross product with $\epsilon$ the volume form of $g$. Note that the last equation of \eqref{1} indicates the absence of charged matter.
The following inequality will be referred to as the charged dominant energy condition
\begin{equation}\label{2}
\mu_{EM}\geq|J_{EM}|_{g}.
\end{equation}

Suppose that $M$ has at least two ends, with one designated end being asymptotically flat, and the remainder being either asymptotically flat or asymptotically cylindrical. Recall that a domain $M_{\text{end}}\subset M$ is an
asymptotically flat end if it is diffeomorphic to $\mathbb{R}^{3}\setminus\text{Ball}$, and in the coordinates given by the asymptotic
diffeomorphism the following fall-off conditions hold
\begin{equation}\label{3}
g_{ij}=\delta_{ij}+o_{l}(r^{-\frac{1}{2}}),\text{ }\text{ }\text{ }\text{ }\partial g_{ij}\in L^{2}(M_{\text{end}}),\text{
}\text{ }\text{ }
\text{ }k_{ij}=O_{l-1}(r^{-\lambda}),\text{
}\text{ }\text{ }
\text{ }\lambda>\frac{5}{2},
\end{equation}
\begin{equation}\label{4}
E^{i}=O_{l-1}(r^{-2}),\text{ }\text{ }\text{ }\text{ }\text{ }B^{i}=O_{l-1}(r^{-2}),
\end{equation}
for some $l\geq 5$\footnote{The notation $f=o_{l}(r^{-a})$ asserts that $\lim_{r\rightarrow\infty}r^{a+n}\partial^{n}f=0$
for all $n\leq l$, and
$f=O_{l}(r^{-a})$ asserts that $r^{a+n}|\partial^{n}f|\leq C$ for all $n\leq l$. The assumption $l\geq 5$ is needed for
the results in \cite{Chrusciel}.}. In the context of the mass-angular momentum-charge inequality, these asymptotics may be weakened, see for example \cite{SchoenZhou}. The asymptotics for cylindrical ends is
most easily described in Brill coordinates, to be given in the next section.

We say that the initial data are axially symmetric if the group of isometries of the Riemannian manifold $(M,g)$ has a subgroup
isomorphic to $U(1)$, and that the remaining quantities defining the initial data are invariant under the $U(1)$ action. In particular, if $\eta$ denotes the Killing field associated with this
symmetry, then
\begin{equation}\label{5}
\mathfrak{L}_{\eta}g=\mathfrak{L}_{\eta}k=\mathfrak{L}_{\eta}E=\mathfrak{L}_{\eta}B=0,
\end{equation}
where $\mathfrak{L}_{\eta}$ denotes Lie differentiation.
If $M$ is simply connected and the data are axially symmetric, it is shown in \cite{Chrusciel} that the analysis reduces to the study
of manifolds of the form
$M\simeq\mathbb{R}^{3}\setminus\sum_{n=1}^{N}i_{n}$, where $i_{n}$ are points in $\mathbb{R}^{3}$ and represent asymptotic ends
(in total there are $N+1$ ends). Moreover there exists a global (cylindrical) Brill coordinate system on $M$, where the points $i_{n}$ all lie on the $z$-axis. Each point represents a black hole, and has the geometry of an asymptotically flat or cylindrical end.
The fall-off conditions in the designated asymptotically flat end guarantee that the ADM mass, ADM angular momentum, and total charges are well-defined by the following limits
\begin{equation}\label{6}
m=\frac{1}{16\pi}\int_{S_{\infty}}(g_{ij,i}-g_{ii,j})\nu^{j},
\end{equation}
\begin{equation}\label{7}
\mathcal{J}=\frac{1}{8\pi}\int_{S_{\infty}}(k_{ij}-(Tr_{g} k)g_{ij})\nu^{i}\eta^{j},
\end{equation}
\begin{equation}\label{8}
q_{e} = \frac1{4\pi} \int_{S_\infty} E_i \nu^i\, , \qquad
q_{b} = \frac1{4\pi} \int_{S_\infty} B_i \nu^i\, ,
\end{equation}
where $S_{\infty}$ indicates the limit as $r\rightarrow\infty$ of integrals over coordinate spheres $S_{r}$, with unit outer normal $\nu$. Here $q_{e}$ and $q_{b}$ denote the total electric and magnetic charge, respectively,
and we denote the square of the total charge by $q^{2}=q_{e}^{2}+q_{b}^{2}$. Note that \eqref{3} implies that the ADM linear momentum vanishes.

In the presence of an electromagnetic field, angular momentum is conserved \cite{Dain}, \cite{DainKhuriWeinsteinYamada} if
\begin{equation}\label{9}
J_{EM}^{i}\eta_{i}=0.
\end{equation}
It is also well-known that angular momentum is conserved if $J^{i}\eta_{i}=0$, where $J=J_{EM}-2E\times B$ is the full momentum density. However, this condition imposes constraints on the angular momentum density of the electromagnetic field, which suggests that it is an undesirable condition in the current setting.
Moreover, when $M$ is simply connected, \eqref{9} is a necessary and sufficient condition \cite{DainKhuriWeinsteinYamada} for the existence of a charged twist potential $\omega$:
\begin{equation}\label{10}
\epsilon_{ijl}(\pi^{ja}+2\gamma^{ja})\eta^{l}\eta_{a}dx^{i}=d\omega
\end{equation}
where
\begin{equation}\label{11}
\pi_{ja}=k_{ja}-(Tr_{g}k)g_{ja}, \qquad  \gamma_{ja}=\epsilon_{cba}E^{c}\epsilon_{j}^{\text{ }\!\text{ }pb}\vec{A}_{p},
\end{equation}
and $\vec{A}$ is a vector potential such that $B=\nabla\times\vec{A}$. However, note that
the topology of $M$ does not allow for a globally defined (smooth) vector potential. The typical construction which avoids this difficulty involves
removing a `Dirac string' associated with each point $i_n$. That is, removing from $M$ either the portion of the $z$-axis below or above $i_{n}$, to obtain
a ($U(1)$ invariant) potential $\vec{A}_{\pm}^{n}$, defined on the complement of the respective Dirac string. We then define
\begin{equation}\label{12}
\vec{A}=\frac{1}{2N}\sum_{n=1}^{N}(\vec{A}_{+}^{n}+\vec{A}_{-}^{n})\text{ }\text{ }\text{ }\text{ on }\text{ }\text{ }\text{ }\mathbb{R}^{3}\setminus\{z-\text{axis}\}.
\end{equation}
Although $\vec{A}$ is discontinuous on the $z$-axis, quantities of the form $\vec{A}\cdot\eta$ remain well behaved since $\eta$ vanishes on the $z$-axis. Nevertheless, the awkward definition makes working with $\vec{A}$ somewhat cumbersome. Fortunately, there is an alternate expression for the twist potential which avoids the use of $\vec{A}$ \cite{ChruscielCosta}, \cite{Costa}, but in general requires reference to an axisymmetric spacetime as opposed to the initial data alone. This is justified, since in electrovacuum the existence of an axisymmetric evolution of the initial data follows from its smoothness \cite{Choquet-Bruhat}, \cite{Chrusciel0}. For our purposes, however, reference to the spacetime can be avoided since the alternate expression may be computed solely from the initial data (see Lemma \ref{lemma5}, and Appendix B).

Let $I_{n}$ denote the
interval of the $z$-axis between $i_{n+1}$ and $i_{n}$, where $i_{0}=-\infty$ and $i_{N+1}=\infty$.
Then a basic formula (Appendix B) yields the angular momentum for each black hole
\begin{equation}\label{13}
\mathcal{J}_{n}=\frac{1}{4}(\omega|_{I_{n}}-\omega|_{I_{n-1}}),
\end{equation}
and conservation of angular momentum shows that the total angular momentum is given by
\begin{equation}\label{14}
\mathcal{J}=\sum_{n=1}^{N}\mathcal{J}_{n}.
\end{equation}

The same heuristic physical argument \cite{ChaKhuri}, \cite{Dain} motivating the mass-angular momentum inequality, leads to the conjectured mass-angular momentum-charge inequality
\begin{equation}\label{15}
m^2\geq\frac{q^{2}+\sqrt{q^{4}+4\mathcal{J}^{2}}}{2},
\end{equation}
whenever there is conservation of angular momentum and charge. In the current setting, such conservation is valid due to the assumptions of axisymmetry, \eqref{9}, and absence of charged matter \eqref{1}. This heuristic derivation is based on the standard picture of gravitational collapse \cite{Choquet-Bruhat}, \cite{ChruscielGallowayPollack}. Thus a counterexample to \eqref{15} would pose serious issues for this model, whereas a verification of \eqref{15} would only lend it, and in particular weak cosmic censorship, more credence. Inequality \eqref{15} has been established by Chru\'{s}ciel and Costa \cite{ChruscielCosta}, \cite{Costa} for electrovacuum initial data under the assumption of simple connectivity, axisymmetry, and maximality ($Tr_{g}k=0$). Under the same hypotheses, Schoen and Zhou \cite{SchoenZhou} relaxed the asymptotics and gave a simplified proof which produced a more detailed lower bound. Their argument also proved the expected rigidity statement, namely, that equality occurs if and only if the initial data arise as the $t=0$ slice of the extreme Kerr-Newman spacetime.

The purpose of this paper is to generalize the results of \cite{ChaKhuri} to include charge, and so the focus here is on the general case without the maximal hypothesis. We also weaken the electrovacuum assumption, and replace it with the charged dominant energy condition \eqref{2}, and \eqref{9}. The presumption of simple connectivity and axisymmetry will be maintained; it should be noted that such results are false without the assumption of axial symmetry \cite{HuangSchoenWang}. In the main result, we exhibit a reduction argument by which the general case is reduced to the maximal case, assuming that a canonical
system of elliptic PDEs possesses a solution. The procedure is modeled on previous reduction arguments that have been applied to other geometric inequalities such as the positive mass theorem and the Penrose inequality \cite{BrayKhuri1}, \cite{BrayKhuri2}, \cite{DisconziKhuri}, \cite{Jang}, \cite{Khuri}, \cite{Khuri1}, \cite{KhuriWeinstein}, \cite{SchoenYau}. Moreover, the primary equation is related to the Jang-type equations that appear in each of
these procedures. The end result yields a natural deformation of the initial data, in which the geometry relevant to the mass-angular momentum-charge inequality is preserved, while achieving the maximal condition.

This paper is organized as follows. In the next section we describe the deformation of $(M,g,k)$ in detail,
while in Section \ref{sec3} the deformation of $(E,B)$ will be constructed. In Section \ref{sec4}
the reduction argument is established and the case of equality is treated. Section \ref{sec5} contains two applications of the deformation procedure, one to the basic inequality yielding a lower bound for the area of black holes in terms of mass, angular momentum, and charge, and another to a weakening of the hypotheses typically assumed in the maximal case of the mass-angular momentum-charge inequality. Lastly three appendices are included to record several important but lengthy calculations.

\section{Deformation of the Metric and Second Fundamental Form}
\label{sec2} \setcounter{equation}{0}
\setcounter{section}{2}

Simple connectivity and axial symmetry imply \cite{Chrusciel} that $M\cong\mathbb{R}^{3}\setminus\sum_{n=1}^{N}i_{n}$, and that there exists a global (cylindrical) Brill coordinate system $(\rho,\phi,z)$ on $M$, where the points $i_{n}$ all lie on the $z$-axis,
and in which the Killing field is given by $\eta=\partial_{\phi}$. In these coordinates the metric takes a simple form
\begin{equation}\label{16}
g=e^{-2U+2\alpha}(d\rho^{2}+dz^{2})+\rho^{2}e^{-2U}(d\phi+A_{\rho}d\rho+A_{z}dz)^{2},
\end{equation}
where $\rho e^{-U}(d\phi+A_{\rho}d\rho+A_{z}dz)$ is the dual 1-form to $|\eta|^{-1}\eta$ and all coefficient functions are independent of $\phi$. Let $M_{end}^{0}$ denote the end associated with limit $r=\sqrt{\rho^{2}+z^{2}}\rightarrow\infty$.
In order to achieve the asymptotically flat fall-off conditions \eqref{3}, it will be assumed that
\begin{equation}\label{17}
U=o_{l-3}(r^{-\frac{1}{2}}),\text{ }\text{ }\text{ }\text{ }\alpha=o_{l-4}(r^{-\frac{1}{2}}),\text{ }\text{ }\text{ }\text{ }A_{\rho},A_{z}=o_{l-3}(r^{-\frac{3}{2}}).
\end{equation}
The remaining ends associated with the points $i_{n}$ will be denoted by $M_{end}^{n}$, and are associated with the limit $r_{n}\rightarrow 0$, where $r_{n}$ is the Euclidean distance to $i_{n}$. The asymptotics for
asymptotically flat and cylindrical ends are given, respectively, by
\begin{equation}\label{18}
U=2\log r_{n}+o_{l-4}(r_{n}^{\frac{1}{2}}),\text{ }\text{ }\text{ }\text{ }\alpha=o_{l-4}(r_{n}^{\frac{1}{2}}),\text{ }\text{ }\text{ }\text{ }A_{\rho},A_{z}=o_{l-3}(r_{n}^{\frac{3}{2}}),
\end{equation}
\begin{equation}\label{19}
U=\log r_{n}+o_{l-4}(r_{n}^{\frac{1}{2}}),\text{ }\text{ }\text{ }\text{ }\alpha=o_{l-4}(r_{n}^{\frac{1}{2}}),\text{ }\text{ }\text{ }\text{ }A_{\rho},A_{z}=o_{l-3}(r_{n}^{\frac{3}{2}}).
\end{equation}
Based on \eqref{3}, \eqref{4}, the corresponding asymptotics for the electromagnetic field and second fundamental form in Brill coordinates are
\begin{equation}\label{18.1}
E_{i}=O_{l}(1),\text{ }\text{ }\text{ }\text{ }E_{\phi}=O_{l}(r_{n}),\text{ }\text{ }\text{ }\text{ }
B_{i}=O_{l}(1),\text{ }\text{ }\text{ }\text{ }B_{\phi}=O_{l}(r_{n}),
\text{ }\text{ }\text{ }i=\rho,z,
\end{equation}
\begin{equation}\label{18.2}
E_{i}=O_{l}(r_{n}^{-1}),\text{ }\text{ }\text{ }\text{ }E_{\phi}=O_{l}(1),\text{ }\text{ }\text{ }\text{ }
B_{i}=O_{l}(r_{n}^{-1}),\text{ }\text{ }\text{ }\text{ }B_{\phi}=O_{l}(1),
\text{ }\text{ }\text{ }i=\rho,z,
\end{equation}
and
\begin{equation}\label{19.1}
k_{\rho\rho}, k_{\rho z}, k_{zz}=O_{l}(r_{n}^{\lambda-4}),\text{ }\text{ }\text{ }\text{ }k_{\rho\phi}, k_{z\phi}=O_{l}(r_{n}^{\lambda-3}),\text{ }\text{ }\text{ }\text{ }k_{\phi\phi}=O_{l}(r_{n}^{\lambda-2}),
\end{equation}
\begin{equation}\label{19.2}
k_{\rho\rho}, k_{\rho z}, k_{zz}=O_{l}(r_{n}^{\lambda-5}),\text{ }\text{ }\text{ }\text{ }k_{\rho\phi}, k_{z\phi}=O_{l}(r_{n}^{\lambda-4}),\text{ }\text{ }\text{ }\text{ }k_{\phi\phi}=O_{l}(r_{n}^{\lambda-3}).
\end{equation}
It will also be assumed that the charged dominant energy condition and the following equation, which is equivalent to \eqref{9}, are satisfied
\begin{equation}\label{20}
div_{g} k(\eta)+2E\times B(\eta)=0.
\end{equation}

We seek a deformation of the initial data $(M,g,k,E,B)\rightarrow(\overline{M},\overline{g},\overline{k},\overline{E},\overline{B})$ such that the manifolds are diffeomorphic $M\cong\overline{M}$, the geometry of the ends is preserved, and
\begin{equation}\label{21}
\overline{m}=m,\text{ }\text{ }\text{ }\text{ }\overline{\mathcal{J}}=\mathcal{J},\text{ }\text{ }\text{ }\text{ }Tr_{\overline{g}}\overline{k}=0,\text{ }\text{ }\text{ }\text{ }
\overline{R}\geq|\overline{k}|_{\overline{g}}^{2}
+2(|\overline{E}|_{\overline{g}}^{2}+|\overline{B}|^{2}_{\overline{g}})\text{ }\text{ weakly,}
\end{equation}
\begin{equation}\label{22}
div_{\overline{g}}\overline{E}=div_{\overline{g}}\overline{B}=0,\text{ }\text{ }\text{ }\text{ }\text{ }\overline{J}_{EM}(\eta)=0,\text{ }\text{ }\text{ }\text{ }\text{ }\overline{q}_{e}=q_{e},\text{ }\text{ }\text{ }\text{ }\text{ }\overline{q}_{b}=q_{b},
\end{equation}
where $\overline{m}$, $\overline{\mathcal{J}}$, $\overline{J}_{EM}$, $\overline{q}_{e}$, $\overline{q}_{b}$, and $\overline{R}$ are the mass, angular momentum, momentum density minus the electromagnetic contribution, charges, and scalar curvature of the new data. The inequality in \eqref{21} is said to hold `weakly' if it is valid when
integrated against an appropriate test function. The validity of this inequality plays a central role in the proof of the mass-angular momentum-charge inequality in the maximal case, and it is precisely the lack of this inequality
in the non-maximal case which prevents the proof from generalizing. Thus, one of the central goals of the deformation is to obtain such a lower bound for the scalar curvature, while preserving all other aspects of the geometry.

The deformation of the metric and second fundamental form is based on that of \cite{ChaKhuri}, and only requires minor modifications. We include the relevant details here, but refer to \cite{ChaKhuri} for the necessary proofs.
Using intuition from previous work \cite{BrayKhuri1}, \cite{BrayKhuri2}, \cite{SchoenYau} it is natural to search for a deformation in the form of a graph inside a stationary 4-manifold
\begin{equation}\label{23}
\overline{M}=\{t=f(x)\}\subset(M\times\mathbb{R}, g+2Y_{i}dx^{i}dt+\varphi dt^{2}),
\end{equation}
where the 1-form $Y=Y_{i}dx^{i}$ and functions $\varphi$ and $f$ are defined on $M$ and satisfy
\begin{equation}\label{24}
\mathfrak{L}_{\eta}f=\mathfrak{L}_{\eta}\varphi=\mathfrak{L}_{\eta}Y=0.
\end{equation}
Set
\begin{equation}\label{25}
\overline{g}_{ij}=g_{ij}+f_{i}Y_{j}+f_{j}Y_{i}+\varphi f_{i}f_{j},\text{ }\text{ }\text{ }\text{ }
\overline{k}_{ij}=\frac{1}{2u}\left(\overline{\nabla}_{i}Y_{j}+\overline{\nabla}_{j}Y_{i}\right),
\end{equation}
where $f_{i}=\partial_{i}f$, $\overline{\nabla}$ is the Levi-Civita connection with respect to $\overline{g}$, and
\begin{equation}\label{26}
u^{2}=\varphi+|Y|_{\overline{g}}^{2}.
\end{equation}
In the `Riemannian' setting \eqref{23}, $\overline{g}$ arises as the induced metric on the graph $\overline{M}$, and in the `Lorentzian' setting
\begin{equation}\label{27}
M=\{t=f(x)\}\subset(\overline{M}\times\mathbb{R}, \overline{g}-2Y_{i}dx^{i}dt-\varphi dt^{2}),
\end{equation}
the deformed data arise as the induced metric and second fundamental form of the $t=0$ slice. Observe that
\begin{equation}\label{28}
\partial_{t}=un-\overline{Y},
\end{equation}
where $n$ is the unit normal to the $t=0$ slice and $\overline{Y}$ is the vector field dual to $Y$ with respect to $\overline{g}$. It follows that $(u,-\overline{Y})$ are the lapse and shift of this stationary spacetime.

Considering the structure of the Kerr-Newman spacetime, we are led to make the following simplifying
assumption that $\overline{Y}$ has only one component
\begin{equation}\label{29}
\overline{Y}^{i}\partial_{i}:=\overline{g}^{ij}Y_{j}\partial_{i}=Y^{\phi}\partial_{\phi}.
\end{equation}
This basic hypothesis ensures that (see \cite{ChaKhuri}) $\overline{g}$ is a Riemannian metric, $Tr_{\overline{g}}\overline{k}=0$, $\varphi=u^{2}-g_{\phi\phi}(Y^{\phi})^{2}$, and if
$\{\overline{e}_{i}\}_{i=1}^{3}$ is an orthonormal frame for $\overline{g}$ with $\overline{e}_{3}=|\eta|^{-1}\eta$, then
\begin{equation}\label{29.1}
\overline{k}(\overline{e}_{i},\overline{e}_{j})=\overline{k}(\overline{e}_{3},\overline{e}_{3})=0,\text{ }\text{ }\text{ }\text{ }
\overline{k}(\overline{e}_{i},\overline{e}_{3})=\frac{|\eta|}{2u}\overline{e}_{i}(Y^{\phi}),
\text{ }\text{ }\text{ }\text{ }i,j\neq 3.
\end{equation}
The deformed data set is then maximal, satisfying one requirement of \eqref{21}. Moreover, this shows that $\varphi$ is determined by the functions
$u$ and $Y^{\phi}$. Hence, the three functions $(u,Y^{\phi},f)$ completely determine the new metric and second fundamental form, and will now be chosen to satisfy the remaining statements in \eqref{21}.
In order that the techniques from the maximal case remain applicable, the existence of a charged twist potential for $(\overline{M},\overline{g},\overline{k},\overline{E},\overline{B})$
is needed. Thus we require
\begin{equation}\label{30}
div_{\overline{g}}\overline{k}(\eta)+2\overline{E}\times\overline{B}(\eta)=0.
\end{equation}
This turns out to be a linear elliptic equation for $Y^{\phi}$ (if $u$ is independent of $Y^{\phi}$); this follows from Appendix B in \cite{ChaKhuri} and Section \ref{sec3} below. From \cite{ChaKhuri} we also deduce that the function $Y^{\phi}$ is uniquely determined
among bounded solutions of \eqref{30}, if the $r^{-3}$-fall-off rate is prescribed at $M_{end}^{0}$. Hence, we will choose the following boundary condition
\begin{equation}\label{31}
Y^{\phi}=-\frac{2\mathcal{J}}{r^{3}}+o_{2}(r^{-\frac{7}{2}})\text{ }\text{ }\text{ as }\text{ }\text{ }r\rightarrow\infty.
\end{equation}
As in Lemma 2.2 of \cite{ChaKhuri}, this guarantees that $\overline{\mathcal{J}}=\mathcal{J}$ if $\overline{g}$ is asymptotically flat and $u\rightarrow 1$ as $r\rightarrow\infty$.

Now let us show how to choose $f$. As with previous deformations associated with the positive mass theorem and Penrose inequality, $f$ is chosen to impart positivity properties to the
scalar curvature. In this regard, it is useful to calculate the scalar curvature for an arbitrary $f$.

\begin{theorem}\label{thm1}
Let $\widetilde{E}$ and $\widetilde{B}$ be arbitrary vector fields on $M$, such that $\widetilde{E}\times\widetilde{B}(\eta)=E\times B(\eta)$, and with the associated matter and momentum densities $\widetilde{\mu}_{EM}$ and $\widetilde{J}_{EM}$. Suppose that \eqref{5}, \eqref{20}, \eqref{24}, \eqref{29}, and \eqref{30} are satisfied, then the scalar curvature of $\overline{g}$ is given by
\begin{align}\label{32}
\begin{split}
\overline{R}-|\overline{k}|_{\overline{g}}^{2}-2(|\overline{E}|_{\overline{g}}^{2}+|\overline{B}|_{\overline{g}}^{2})
=& 16\pi(\widetilde{\mu}_{EM}-\widetilde{J}_{EM}(v))+|k-\pi|_{g}^{2}+2u^{-1}div_{\overline{g}}(uQ)\\
&+(Tr_{g}\pi)^{2}-(Tr_{g}k)^{2}+2v(Tr_{g}\pi-Tr_{g}k)\\
&+2(|\widetilde{E}|_{g}^{2}+|\widetilde{B}|_{g}^{2}+2\widetilde{E}\times \widetilde{B}(v)-|\overline{E}|_{\overline{g}}^{2}-|\overline{B}|_{\overline{g}}^{2})\\
&+4u^{-1}Y^{\phi}\left(\frac{\widetilde{E}\times \widetilde{B}(\eta)}{\sqrt{1+u^{2}|\nabla f|_{g}^{2}}}-\overline{E}\times
\overline{B}(\eta)\right),
\end{split}
\end{align}
where
\begin{equation}\label{33}
\pi_{ij}=\frac{u\nabla_{ij}f+u_{i}f_{j}+u_{j}f_{i}+\frac{1}{2u}(g_{i\phi}Y^{\phi}_{,j}+g_{j\phi}Y^{\phi}_{,i})}{\sqrt{1+u^{2}|\nabla f|_{g}^{2}}}
\end{equation}
is the second fundamental form of the graph $M$ in the Lorentzian setting,
\begin{equation}\label{34}
v^{i}=\frac{uf^{i}}{\sqrt{1+u^{2}|\nabla f|_{g}^{2}}},\text{ }\text{ }\text{ }\text{ }w^{i}=\frac{uf^{i}+u^{-1}\overline{Y}^{i}}{\sqrt{1+u^{2}|\nabla f|_{g}^{2}}},
\end{equation}
\begin{equation}\label{35}
Q_{i}=\overline{Y}^{j}\overline{\nabla}_{ij}f-u\overline{g}^{jl}f_{l}\overline{k}_{ij}+w^{j}(k-\pi)_{ij}+uf_{i}w^{l}w^{j}(k-\pi)_{lj}\sqrt{1+u^{2}|\nabla f|_{g}^{2}},
\end{equation}
and $f_{i}=\partial_{i}f$, $f^{i}=g^{ij}f_{j}$. Furthermore, if $Y\equiv 0$ then the same conclusion holds without any of the listed hypotheses.
\end{theorem}

\begin{proof}
In Appendix A of \cite{ChaKhuri} it is shown that
\begin{align}\label{35.1}
\begin{split}
\overline{R}-|\overline{k}|_{\overline{g}}^{2}=& 16\pi(\mu-J(v))+|k-\pi|_{g}^{2}+\frac{2}{u}div_{\overline{g}}(uQ)\\
&+(Tr_{g}\pi)^{2}-(Tr_{g}k)^{2}+2v(Tr_{g}\pi-Tr_{g}k)\\
&-\frac{2}{u}\left(\frac{div_{g}k(\byp \eta)}{\sqrt{1+u^{2}|\nabla f|_{g}^{2}}} - div_{\bg}\bk(\byp \eta)\right),
\end{split}
\end{align}
where $\mu$ and $J$ are the full matter and momentum densities. Since
\begin{equation}\label{35.2}
16\pi\mu=16\pi\widetilde{\mu}_{EM}+2(|\widetilde{E}|_{g}^{2}+|\widetilde{B}|_{g}^{2}),\text{ }\text{ }\text{ }\text{ }\text{ }8\pi J=8\pi\widetilde{J}_{EM}-2\widetilde{E}\times\widetilde{B},
\end{equation}
and
\begin{equation}\label{35.3}
div_{g}k(\eta)=-2E\times B(\eta)=-2\widetilde{E}\times\widetilde{B}(\eta),\text{ }\text{ }\text{ }\text{ }\text{ }
div_{\bg}\bk(\eta)=-2\overline{E}\times\overline{B}(\eta),
\end{equation}
the desired result follows.
\end{proof}

In the next section, the deformation of the electromagnetic field will be described and will consist of two steps, namely $(E,B)\rightarrow(\widetilde{E},\widetilde{B})$ and $(\widetilde{E},\widetilde{B})\rightarrow(\overline{E},\overline{B})$. Moreover, the deformation will be chosen so that the charged dominant energy condition with respect to $(\widetilde{E},\widetilde{B})$ follows from \eqref{2}, and
\begin{equation}\label{36}
|\overline{E}|_{\overline{g}}^{2}+|\overline{B}|_{\overline{g}}^{2}
=|\widetilde{E}|_{g}^{2}+|\widetilde{B}|_{g}^{2}+2\widetilde{E}\times \widetilde{B}(v),
\end{equation}
\begin{equation}\label{37}
\overline{E}\times\overline{B}(\eta)=\frac{\widetilde{E}\times \widetilde{B}(\eta)}{\sqrt{1+u^{2}|\nabla f|_{g}^{2}}}.
\end{equation}
Thus, this theorem makes it clear that in order to obtain the inequality in \eqref{21} at least weakly, $f$ should be chosen to solve the equation
\begin{equation}\label{38}
Tr_{g}(\pi-k)=0.
\end{equation}
It follows that
\begin{equation}\label{39}
\overline{R}-|\overline{k}|_{\overline{g}}^{2}-2(|\overline{E}|_{\overline{g}}^{2}+|\overline{B}|_{\overline{g}}^{2})= 16\pi(\widetilde{\mu}_{EM}-\widetilde{J}_{EM}(v))+|k-\pi|_{g}^{2}+2u^{-1}div_{\overline{g}}(uQ),
\end{equation}
which yields the desired weak inequality after multiplying by $u$ and applying the divergence theorem; it is assumed that appropriate asymptotic conditions are imposed (see below) in order
to ensure that the boundary integrals vanish in each of the ends. Equation \eqref{37} is the same equation which played a central role in \cite{ChaKhuri}, and is similar to previous Jang-type equations utilized in connection with the positive mass theorem \cite{SchoenYau} and the Penrose inequality \cite{BrayKhuri2}.
As discussed in \cite{ChaKhuri}, one primary difference however, is that solutions of \eqref{38} do not blow-up, while solutions of \eqref{39}
typically blow-up at apparent horizons or can be prescribed to blow-up at these surfaces \cite{HanKhuri}. Lastly, in order to ensure that $\overline{m}=m$, we will impose the following asymptotics
\begin{equation}\label{40}
|f|+r|\nabla f|_{g}+r^{2}|\nabla^{2} f|_{g}\leq cr^{-\varepsilon}\text{ }\text{ }\text{ in }\text{ }\text{ }M_{end}^{0},
\end{equation}
for some $0<\varepsilon<1$. It is shown in \cite{ChaKhuri} that a bounded solution exists (given $u$) by prescribing the following asymptotics at the remaining ends
\begin{equation}\label{41}
r_{n}^{-1}|\nabla f|_{g}+r_{n}^{-2}|\nabla^{2} f|_{g}\leq c\text{ }\text{ }\text{ in asymptotically flat }\text{ }\text{ }M_{end}^{n},
\end{equation}
\begin{equation}\label{42}
|\nabla f|_{g}+|\nabla^{2} f|_{g}\leq cr_{n}^{\frac{1}{2}}\text{ }\text{ }\text{ in asymptotically cylindrical }\text{ }\text{ }M_{end}^{n}.
\end{equation}

We have now shown how to choose $f$ and $Y$, in order to produce a deformation of the metric and second fundamental form which satisfies \eqref{21}, assuming that the deformation of the electromagnetic field (given in the next section) is chosen appropriately. It will then remain to choose $u$, in such a way as to facilitate a proof of the mass-angular momentum-charge inequality. This shall be accomplished in Section \ref{sec4}.

\section{Deformation of the Electric and Magnetic Fields}
\label{sec3} \setcounter{equation}{0}
\setcounter{section}{3}

The deformation of the Maxwell field will be chosen to satisfy \eqref{22}, \eqref{36}, and \eqref{37}, and will consist of two steps. Motivation for the definition comes from the case of equality in \eqref{15}. In this case the Lorentzian manifold \eqref{27} should be the extreme Kerr-Newman spacetime, and the graph $M=\{t=f(x)\}$ should give the desired embedding of the initial data. Thus $(\overline{E},\overline{B})$ must be the Maxwell field of the extreme Kerr-Newman solution induced on the $t=0$ slice. This field satisfies
\begin{equation}\label{43}
\overline{E}(\eta)=\overline{B}(\eta)=0.
\end{equation}
It turns out that condition \eqref{43} is also highly important for establishing \eqref{22}, \eqref{36}, and \eqref{37}. On the other hand $(E,B)$ should arise as the induced Maxwell field on the graph from the field strength tensor associated with $(\overline{E},\overline{B})$ (see below). This yields a relation in which $(\overline{E},\overline{B})$ is expressed in terms of $(E,B)$. Unfortunately however, the relation is not consistent with \eqref{43} unless $(E,B)$ satisfies a certain condition \eqref{47} a priori (the condition \eqref{47} is of course valid in the extreme Kerr-Newman setting). This suggests that an initial deformation $(E,B)\rightarrow(\widetilde{E},\widetilde{B})$ should be performed in order to achieve \eqref{47}, before expressing $(\overline{E},\overline{B})$ in terms of $(\widetilde{E},\widetilde{B})$ through the field strength relation.

We now describe the initial deformation $(E,B)\rightarrow(\widetilde{E},\widetilde{B})$. Let
\begin{equation}\label{44}
e_{1}=e^{U-\alpha}(\partial_{\rho}-A_{\rho}\partial_{\phi}),
\text{ }\text{ }\text{ }\text{ } e_{2}=e^{U-\alpha}(\partial_{z}-A_{z}\partial_{\phi}),
\text{ }\text{ }\text{ }\text{ } e_{3}=\frac{1}{\sqrt{g_{\phi\phi}}}\partial_{\phi},
\end{equation}
be an orthonormal frame for $(M,g)$. Then set
\begin{equation}\label{45}
\widetilde{E}(e_{i})=E(e_{i}), \text{ }\text{ }\text{ }\text{ }\text{ }  \widetilde{B}(e_{i})= B(e_{i}) \text{ }\text{ }\text{ }\text{ for }\text{ }\text{ }\text{ } i= 1, 2,
\end{equation}
\begin{equation}\label{46}
\widetilde{E}(e_{3})=v \times B(e_{3}),\text{ }\text{ }\text{ }\text{ }\text{ }
\widetilde{B}(e_{3})=-v \times E(e_{3}),
\end{equation}
where $v$ is given in \eqref{34}. Notice that $(\widetilde{E},\widetilde{B})$ agrees with $(E,B)$ except in the $\phi$-component. In particular, since $v(e_{3})=0$ the expressions $v \times B(e_{3})$ and $v \times E(e_{3})$ are independent of $E(e_{3})$ and $B(e_{3})$, and hence
\begin{equation}\label{47}
\widetilde{E}(e_{3})=v \times \widetilde{B}(e_{3}),\text{ }\text{ }\text{ }\text{ }\text{ }
\widetilde{B}(e_{3})=-v \times \widetilde{E}(e_{3}).
\end{equation}
As will be shown in Lemma \ref{lemma2} below, this condition is precisely what is needed to ensure that \eqref{43} is valid. But first let us record the properties of this initial deformation.

\begin{lemma}\label{lemma1}
If \eqref{4}, \eqref{5}, and \eqref{24} hold, then
\begin{equation}\label{48}
\widetilde{q}_{e}=q_{e},\text{ }\text{ }\text{ }\text{ }\text{ }\text{ }\widetilde{q}_{b}=q_{b},
\end{equation}
\begin{equation}\label{49}
div_{g} \widetilde{E} = div_{g} E , \text{ }\text{ }\text{ }\text{ }\text{ }\text{ }
div_{g}\widetilde{B} =div_{g}B,
\end{equation}
\begin{equation}\label{50}
\widetilde{J}_{EM}(\eta) = J_{EM}(\eta),\text{ }\text{ }\text{ }\text{ }\text{ }\text{ }
\widetilde{E}\times\widetilde{B}(\eta)=E\times B(\eta),
\end{equation}
\begin{equation}\label{51}
8 \pi(\widetilde{\mu}_{EM}-\widetilde{J}_{EM}(v)) - 8 \pi(\mu_{EM}-J_{EM}(v))
=\left(E(e_{3})- \widetilde{E}(e_{3})\right)^{2}  + \left(B(e_{3})- \widetilde{B}(e_{3}) \right)^{2}.
\end{equation}
\end{lemma}

\begin{proof}
Recall the expression for $q_{e}$ in \eqref{8}. The normal derivative is given by $\nu=|\partial_{r}|^{-1}
\partial_{r}$, where $\partial_{r}=\sin\theta\partial_{\rho}+\cos\theta\partial_{z}$ arises from the change of cylindrical to spherical coordinates $\rho=r\sin\theta$, $z=r\cos\theta$. In particular, $\nu$ does not have a $\phi$-component so that $\widetilde{E}\cdot\nu=E\cdot\nu$ and $\widetilde{B}\cdot\nu=B\cdot\nu$. It follows that $\widetilde{q}_{e}=q_{e}$ and $\widetilde{q}_{b}=q_{b}$.

Next observe that \eqref{5} implies
\begin{align}\label{52}
\begin{split}
div_{g}E&=\sum_{i=1}^{3}g(\nabla_{e_{i}}E,e_{i})\\
&= \sum_{i=1}^{3}e_{i}\left[E(e_{i})\right]-g(E,\nabla_{e_{i}}e_{i})\\
&=\sum_{i=1}^{2}e_{i}\left[E(e_{i})\right]-\sum_{i=1}^{3}\sum_{j=1}^{2}g(e_{j},\nabla_{e_{i}}e_{i})E(e_{j}),
\end{split}
\end{align}
since
\begin{equation}\label{53}
g(e_{\phi},\nabla_{e_{i}}e_{i})=-|\eta|^{-1}g(\nabla_{e_{i}}\eta,e_{i})=0
\end{equation}
as $\eta$ is a Killing field. This shows that $div_{g}E$ may be computed independent of $E(e_{3})$, and the same holds for $div_{g}\widetilde{E}$. Identities \eqref{49} now follow, since $(E,B)$ and $(\widetilde{E},\widetilde{B})$ agree except in the $\phi$-component. This level of agreement also yields
\begin{equation}\label{54}
8\pi(J_{EM}(\eta) - \widetilde{J}_{EM}(\eta))
=2E\times B(\eta) - 2 \widetilde{E} \times \widetilde{B}(\eta) =0.
\end{equation}

Lastly, a direct computation using $v(e_{3})=0$ produces \eqref{51}. Namely
\begin{align}\label{55}
\begin{split}
&8 \pi(\widetilde{\mu}_{EM}-\widetilde{J}_{EM}(v)) - 8 \pi(\mu_{EM}-J_{EM}(v)) \\
=&  |E|_{g}^{2} + |B|_{g}^{2} + 2E \times B(v) - |\widetilde{E}|_{g}^{2} - |\widetilde{B}|_{g}^{2} - 2\widetilde{E} \times \widetilde{B}(v) \\
=& E(e_{3})^{2} + B(e_{3})^{2} - \widetilde{E}(e_{3})^{2} - \widetilde{B}(e_{3})^{2} \\
&+ 2[v(e_{1})(E(e_{2})B(e_{3})-E(e_{3})B(e_{2}))+v(e_{2})(E(e_{3})B(e_{1})- E(e_{1})B(e_{3}))] \\
&-2[v(e_{1})(\widetilde{E}(e_{2})\widetilde{B}(e_{3})-\widetilde{E}(e_{3})\widetilde{B}(e_{2})) +v(e_{2})(\widetilde{E}(e_{3})\widetilde{B}(e_{1})- \widetilde{E}(e_{1})\widetilde{B}(e_{3}))]\\
=&E(e_{3})^{2} + B(e_{3})^{2} - \widetilde{E}(e_{3})^{2} - \widetilde{B}(e_{3})^{2} \\
&+ 2[E(e_{3})(v(e_{2})B(e_{1})-v(e_{1})B(e_{2}))+B(e_{3})(v(e_{1}) E(e_{2})-v(e_{2})E(e_{1}))] \\
&-2[\widetilde{E}(e_{3})(v(e_{2})B(e_{1})-v(e_{1})B(e_{2}))+\widetilde{B}(e_{3})( v(e_{1}) E(e_{2})-v(e_{2})E(e_{1}))] \\
=&E(e_{3})^{2} + B(e_{3})^{2} - \widetilde{E}(e_{3})^{2} - \widetilde{B}(e_{3})^{2}
-2 \left(E(e_{3})\widetilde{E}(e_{3}) + B(e_{3})\widetilde{E}(e_{3})\right)
+2\left( \widetilde{E}(e_{3})^{2}+\widetilde{B}(e_{3})^{2}\right)   \\
=& \left(E(e_{3})- \widetilde{E}(e_{3})\right)^{2}  + \left(B(e_{3})- \widetilde{B}(e_{3}) \right)^{2}.
\end{split}
\end{align}
\end{proof}

This lemma shows that the initial deformation preserves all relevant quantities and the charged dominant energy condition. Moreover it satisfies the relation \eqref{47}, which will be shown to guarantee the important property \eqref{43}. In order to see this, we must now define the second step in the deformation $(\widetilde{E},\widetilde{B})\rightarrow(\overline{E},\overline{B})$.

Assuming that the functions $(u,Y^{\phi},f)$ are chosen to possess the appropriate asymptotics, the geometry of the ends will be preserved in the deformation. Since the deformed
data are also simply connected and axially symmetric, the results of \cite{Chrusciel} apply to yield a global Brill coordinate system $(\overline{\rho},\overline{z},\phi)$ such that
\begin{equation}\label{56}
\overline{g}=e^{-2\overline{U}+2\overline{\alpha}}(d\overline{\rho}^{2}+d\overline{z}^{2})
+\overline{\rho}^{2}e^{-2\overline{U}}(d\phi+A_{\overline{\rho}}d\overline{\rho}+A_{\overline{z}}d\overline{z})^{2}.
\end{equation}
Let
\begin{equation}\label{57}
\overline{e}_{1}=e^{\overline{U}-\overline{\alpha}}(\partial_{\overline{\rho}}-A_{\overline{\rho}}\partial_{\phi}),
\text{ }\text{ }\text{ }\text{ } \overline{e}_{2}=e^{\overline{U}-\overline{\alpha}}(\partial_{\overline{z}}-A_{\overline{z}}\partial_{\phi}),
\text{ }\text{ }\text{ }\text{ } \overline{e}_{3}=\frac{1}{\sqrt{\overline{g}_{\phi\phi}}}\partial_{\phi}
=\frac{1}{\sqrt{g_{\phi\phi}}}\partial_{\phi}=e_{3},
\end{equation}
be the corresponding orthonormal frame, where $\overline{g}_{\phi\phi}=g_{\phi\phi}$ follows from \eqref{24}, and let
\begin{equation}\label{57.1}
\overline{\theta}^{1}=e^{-\overline{U}+\overline{\alpha}}d\overline{\rho},
\text{ }\text{ }\text{ }\text{ }
\overline{\theta}^{2}=e^{-\overline{U}+\overline{\alpha}}d\overline{z},
\text{ }\text{ }\text{ }\text{ }
\overline{\theta}^{3}=\overline{\rho}e^{-\overline{U}}(d\phi+A_{\overline{\rho}}d\overline{\rho}+A_{\overline{z}}d\overline{z})
\end{equation}
be the dual co-frame. Consider the field strength tensor
\begin{equation}\label{58}
\overline{F}=\frac{1}{2}\overline{F}_{ab}dx^{a} \wedge dx^{b}
\end{equation}
on the spacetime $(\overline{M}\times\mathbb{R}, \overline{g}-2Y_{i}dx^{i}dt-\varphi dt^{2})$, such that
\begin{equation} \label{59}
\overline{F}(\overline{e}_{i},n)= \overline{E}_{i}, \text{ }\text{ }\text{ }\text{ }\text{ }
\overline{F}(\overline{e}_{i},\overline{e}_{j})= \overline{\epsilon}_{ijl}\overline{B}^{l}, \text{ }\text{ }\text{ }\text{ }\text{ } i,j,l = 1,2,3,
\end{equation}
with $\overline{\epsilon}$ the volume form of $\overline{g}$, and
$n=u^{-1}(\partial_{t}+\overline{Y})$ the unit normal of the $t=0$ slice.
Here indices are raised and lowered with respect to $\overline{g}$, and
\begin{equation}\label{60}
\overline{E}=\sum_{i=1}^{3}\overline{E}^{i}\overline{e}_{i}, \text{ }\text{ }\text{ }\text{ }\text{ }
\overline{B}=\sum_{i=1}^{3}\overline{B}^{i}\overline{e}_{i}.
\end{equation}
Thus $(\overline{E},\overline{B})$ arise as the induced electric and magnetic field on the $t=0$ slice, and with respect to the basis $\{n,\overline{e}_{i}\}$ we have
\begin{equation}\label{61}
\overline{F}=\left(\bF_{ij}\right) =\left( \begin{array}{cccc} 0 & -\overline{E}_{1} & -\overline{E}_{2} & -\overline{E}_{3} \\
\overline{E}_{1} & 0 & \overline{B}_{3} & -\overline{B}_{2} \\
\overline{E}_{2} & -\overline{B}_{3} & 0 & \overline{B}_{1} \\
\overline{E}_{3} &  \overline{B}_{2} & -\overline{B}_{1} & 0
\end{array} \right).
\end{equation}

The field strength $\overline{F}$ will be defined from the given fields $(\widetilde{E},\widetilde{B})$ in the following way. Let
\begin{equation}\label{62}
N= \frac{u\barna f + n}{\sqrt{\volbarg}}, \text{ }\text{ }\text{ }\text{ }\text{ }
X_{i}=\overline{e}_{i}+\overline{e}_{i}(f)\partial_{t}, \text{ }\text{ }\text{ }\text{ }i=1,2,3,
\end{equation}
be the unit normal and a basis of tangent vectors for the graph $(M=\{t=f(x)\},g,\pi)$ sitting inside the spacetime, where $\overline{\nabla}^{i}f=\overline{g}^{ij}f_{j}$. Then set
\begin{equation}\label{63}
\overline{F}(X_{i}, N) = \widetilde{E}_{i}, \text{ }\text{ }\text{ }\text{ }\text{ }
\overline{F} (X_{i}, X_{j}) = \epsilon_{ijl} \widetilde{B}^{l},
\text{ }\text{ }\text{ }\text{ }\text{ } i,j,l = 1,2,3.
\end{equation}
Here indices are raised and lowered with respect to $g$, and
\begin{equation}\label{64}
\widetilde{E}=\sum_{i=1}^{3}\widetilde{E}^{i}\overline{e}_{i}, \text{ }\text{ }\text{ }\text{ }\text{ }
\widetilde{B}=\sum_{i=1}^{3}\widetilde{B}^{i}\overline{e}_{i}.
\end{equation}
Thus $(\widetilde{E},\widetilde{B})$ arise as the induced electric and magnetic field on the graph $(M=\{t=f(x)\},g,\pi)$. The equations \eqref{59} and \eqref{63} yield relations between $(\overline{E},\overline{B})$ and $(\widetilde{E},\widetilde{B})$, which we now describe.

\begin{lemma} \label{lemma2}
For $i=1,2$,
\begin{align}\label{65}
\begin{split}
\be^{i}
&= \frac{\widetilde{E}^{i}}{\sqrt{\volg}} - \sqrt{\volg}\left(v \times \left(\widetilde{B} + v \times \widetilde{E}\right) \right)^{i},\\
\bm^{i}
&= \frac{\widetilde{B}^{i}}{\sqrt{\volg}} + \sqrt{\volg}\left(v \times \left( \widetilde{E} - v \times \widetilde{B}\right) \right)^{i},  \\
\be_{3}&= \sqrt{\volg} \left(\widetilde{E}_{3} - v \times \widetilde{B}(e_{3})\right), \\
\bm_{3}&= \sqrt{\volg} \left(\widetilde{B}_{3} + v \times \widetilde{E}(e_{3})\right).
\end{split}
\end{align}
\end{lemma}

\begin{proof}
Let $g_{ij}=g(\overline{e}_{i},\overline{e}_{j})$, with $g^{ij}$ components of the inverse matrix, and recall \cite{ChaKhuri} that
\begin{equation} \label{66}
g_{ij}= \bg_{ij} - f_{i}Y_{j} - f_{j}Y_{i} - \varphi f_{i}f_{j},
\text{ }\text{ }\text{ }\text{ }\text{ }g^{ij}= \bg^{ij} - u^{-2}\by^{i}\by^{j} + w^{i}w^{j},
\end{equation}
where
\begin{equation} \label{67}
w^{i}= \frac{u\overline{g}^{ij}f_{j} + u^{-1}\by^{i}}{\sqrt{\volbarg}},
\text{ }\text{ }\text{ }\text{ }\text{ }u^{2} = \varphi + | \by |_{\bg}^{2}.
\end{equation}
Then a direct calculation yields
\begin{align}\label{68}
\begin{split}
\widetilde{E}^{1} =&  g^{1i} \bF(X_{i}, N) \\
=& \left( \delta^{1i} + w(\overline{\theta}^{1}) w(\overline{\theta}^{i})  \right)
\bF \left( \overline{e}_{i} + \overline{e}_{i}(f)(u n - \by), \frac{u \barna f + n}{\sqrt{\volbarg}}   \right) \\
=& \frac{\left( \delta^{1i} + w(\overline{\theta}^{1}) w(\overline{\theta}^{i})  \right)}{\sqrt{\volbarg}}
\left[ u \overline{e}_{j}(f) \bF_{i}^{\text{ }\text{ }\!j}
+ \be_{i}
- \overline{e}_{i}(f) \left(   u^{2}\be(\barna f)
+ u \overline{e}_{j}(f)\by(\overline{e}_{3}) \bF_{3}^{\text{ }\text{ }\!j}
+ \by(\overline{e}_{3}) \be_{3}
\right)\right]  \\
=& \frac{\be^{1} + u \overline{e}_{2}(f) \bm_{3}}{\sqrt{\volbarg}}
- \frac{\overline{e}_{1}(f)}{\sqrt{\volbarg}} \left(u^{2}\be(\barna f)
+  u \overline{e}_{j}(f)\by(\overline{e}_{3}) \bF_{3}^{\text{ }\text{ }\!j}
+  \by(\overline{e}_{3})\be_{3}\right)  \\
&+ \frac{u^{2} \overline{e}_{1}(f)}{(\volbarg)^{3/2}} \left((\volbarg) \be (\barna f)
 - u |\barna f|_{\bg}^{2} \overline{e}_{j}(f)\by(\overline{e}_{3}) \bF_{3}^{\text{ }\text{ }\!j}
-|\barna f|_{\bg}^{2}  \by(\overline{e}_{3})\be_{3}\right) \\
&+ \frac{\overline{e}_{1}(f) \by(\overline{e}_{3})}{(\volbarg)^{3/2}} \left(
u \overline{e}_{j}(f) \bF_{3}^{\text{ }\text{ }\!j}
+ \be_{3}
\right) \\
=& \frac{\be^{1}+u \overline{e}_{2}(f) \bm_{3}}{\sqrt{\volbarg}}.
\end{split}
\end{align}
Likewise,
\begin{equation}\label{69}
\widetilde{E}^{2} = \frac{\be^{2} - u \overline{e}_{1}(f) \bm_{3}}{\sqrt{\volbarg}}.
\end{equation}
A similar computation yields the remaining component of $\widetilde{E}$, namely
\begin{align} \label{70}
\begin{split}
\widetilde{E}^{3} =&  g^{3i} \bF(X_{i}, N) \\
=& \left( \delta^{3i} -\frac{|\by|_{\bg}^{2} \delta^{3i}}{u^{2}} + w(\overline{\theta}^{3})w(\overline{\theta}^{i})  \right) \bF \left(\overline{e}_{i} + \overline{e}_{i}(f)(u n - \by), \frac{u \barna f + n}{\sqrt{\volbarg}}\right) \\
=&\frac{\left( 1- u^{-2}|\by|_{\bg}^{2} \right)}{\sqrt{\volbarg}}
\left( u \overline{e}_{j}(f) \bF^{3j} + \be^{3} \right) \\
&+ \frac{\by(\overline{e}_{3})}{(\volbarg)^{3/2}}
\left((\volbarg)\be(\barna f)-|\barna f|_{\bg}^{2}
\left(u\overline{e}_{j}(f)\by(\overline{e}_{3}) \bF^{3j}
+ \by(\overline{e}_{3}) \be^{3}\right)\right) \\
&+ \frac{|\by|_{\bg}^{2}}{u^{2}(\volbarg)^{3/2}}
\left(u \overline{e}_{j}(f) \bF^{3j}+\be^{3}\right) \\
=& \frac{\be^{3}+u \left(\overline{e}_{1}(f) \bm_{2}- \overline{e}_{2}(f) \bm_{1}\right) + \by(\overline{e}_{3})\be(\barna f)}{\sqrt{\volbarg}}.
\end{split}
\end{align}

We will now compute each component of $\widetilde{B}$. Recall that
$\epsilon$ and $\overline{\epsilon}$ are the volume forms of $g$ and $\bg$ respectively, which satisfy the following relation (Lemma 2.1 of \cite{ChaKhuri})
\begin{equation} \label{71}
\epsilon = \sqrt{\volbarg} \text{ } \overline{\epsilon}.
\end{equation}
Therefore
\begin{equation} \label{72}
\widetilde{B}^{1} =  \frac{\bF(X_{2}, X_{3})}{\sqrt{\volbarg}}
= \frac{\bF \left(\overline{e}_{2} + \overline{e}_{2}(f)(u n - \by), \overline{e}_{3}   \right)}{\sqrt{\volbarg}}
=\frac{ \bm^{1} - u \overline{e}_{2}(f) \be_{3}}{\sqrt{\volbarg}},
\end{equation}
\begin{equation} \label{73}
\widetilde{B}^{2} =  \frac{\bF(X_{3}, X_{1})}{\sqrt{\volbarg}}
= \frac{\bF \left( \overline{e}_{3}, \overline{e}_{1} + \overline{e}_{1}(f)(u n - \by)\right)}{\sqrt{\volbarg}}
= \frac{\bm^{2} + u \overline{e}_{1}(f) \be_{3}}{\sqrt{\volbarg}},
\end{equation}
and
\begin{equation} \label{74}
\begin{split}
\widetilde{B}^{3} &=  \frac{\bF(X_{1}, X_{2})}{\sqrt{\volbarg}} \\
&= \frac{1}{\sqrt{\volbarg}}\bF \left( \overline{e}_{1} + \overline{e}_{1}(f)(u n - \by), \overline{e}_{2} + \overline{e}_{2}(f)(u n - \by)\right) \\
&=\frac{1}{\sqrt{\volbarg}}  \left( \bm^{3} + \by(e_{3}) \bm(\barna f) + u \left(e_{2}(f)\be_{1} - e_{1}(f) \be_{2} \right)  \right).
\end{split}
\end{equation}

We will now solve for $(\be,\bm)$ in terms of $(\widetilde{E},\widetilde{B})$. In order to do this, the following identities, derived from \eqref{68}-\eqref{70} and \eqref{72}-\eqref{74}, will be used:
\begin{equation}\label{75}
\widetilde{E}(\nabla f) = \widetilde{E}^{1} \overline{e}_{1}(f) + \widetilde{E}^{2}\overline{e}_{2}(f) = \frac{\be(\barna f)}{\sqrt{\volbarg}},
\end{equation}
\begin{equation}\label{76}
\widetilde{B} (\nabla f) = \widetilde{B}^{1} \overline{e}_{1}(f) + \widetilde{B}^{2}\overline{e}_{2}(f) = \frac{\bm(\barna f)}{\sqrt{\volbarg}},
\end{equation}
\begin{equation}\label{77}
\begin{split}
 (\nabla f \times \widetilde{E})(\overline{e}_{3})
&= \epsilon_{ij3}f^{i}\widetilde{E}^{j}   \\
&= \sqrt{\volbarg}\text{ }\overline{\epsilon}_{ij3}\left(\frac{\barna f+ |\barna f|_{\bg}^{2} \by}{\volbarg} \right)^{i}\widetilde{E}^{j} \\
&= \frac{(\barna f \times \overline{E})(\overline{e}_{3}) - u |\barna f|_{\bg}^{2} \overline{B}_{3}}{\volbarg},
\end{split}
\end{equation}
and similarly
\begin{equation}\label{78}
 (\nabla f \times \widetilde{B})(\overline{e}_{3})
= \epsilon_{ij3}f^{i}\widetilde{B}^{j}
= \frac{(\barna f \times \overline{B})(\overline{e}_{3}) + u |\barna f|_{\bg}^{2} \overline{E}_{3}}{\volbarg}.
\end{equation}
Consider the third component of $(\widetilde{E},\widetilde{B})$. By substituting the identities above and utilizing $\overline{e}_{3}=e_{3}$, the definition of $v$ \eqref{34}, as well as
\begin{equation}\label{79}
(1+u^{2}|\nabla f|_{g}^{2})(1-u^{2}|\overline{\nabla}f|_{\overline{g}}^{2})=1
\end{equation}
from Lemma 2.1 of \cite{ChaKhuri}, we obtain
\begin{align} \label{80}
\begin{split}
\widetilde{E}_{3}= \widetilde{E}^{i}g_{i3}&= \widetilde{E}^{i}\left(\delta_{i3} -\overline{e}_{i}(f)\by(\overline{e}_{3}) \right)\\
&=\widetilde{E}^{3}-\widetilde{E}(\nabla f)\overline{Y}(\overline{e}_{3})\\
&= \frac{\be^{3}+ u(\overline{e}_{1}(f)\bm_{2}-\overline{e}_{2}(f)\bm_{1})}{\sqrt{\volbarg}}\\
&=\frac{\be_{3}+ u\left[(\volbarg)\nabla f\times\widetilde{B}(\overline{e}_{3})
-u|\overline{\nabla}f|_{\overline{g}}^{2}\overline{E}_{3}\right]}{\sqrt{\volbarg}}\\
&= \sqrt{\volbarg}\text{ }\overline{E}_{3} + v \times \widetilde{B}(e_{3}),
\end{split}
\end{align}
and likewise
\begin{equation} \label{81}
\widetilde{B}_{3}=\sqrt{\volbarg}\text{ }\overline{B}_{3} - v \times \widetilde{E}(e_{3}).
\end{equation}
Therefore
\begin{equation} \label{82}
\be_{3}= \sqrt{\volg} \left(\widetilde{E}_{3} - v \times \widetilde{B}(e_{3})\right),
\end{equation}
\begin{equation}\label{83}
\bm_{3}= \sqrt{\volg} \left(\widetilde{B}_{3} + v \times \widetilde{E}(e_{3})\right).
\end{equation}

Next, notice that for any vector field $Z$ on $(M,g)$
\begin{align} \label{84}
\begin{split}
(v \times Z)^{1} &= \epsilon_{ijl}g^{1i}v^{j}Z^{l} \\
&=\overline{\epsilon}_{ijl} \left( \delta^{1i} + w(\overline{\theta}^{1})w(\overline{\theta}^{i}) \right)
\left( u\overline{\nabla}^{j}f+  u |\barna f |_{\bg}^{2}\by^{j} \right) Z^{l} \\
&= u \overline{e}_{2}(f) \left(Z^{3}- \by(\overline{e}_{3})Z(\nabla f) \right)\\
&= u \overline{e}_{2}(f)Z_{3},
\end{split}
\end{align}
and similarly
\begin{equation} \label{85}
(v \times Z)^{2} = -u\overline{e}_{1}(f)Z_{3}.
\end{equation}
Hence, we may solve for $\be^{1}$ from \eqref{68} by employing \eqref{83} and \eqref{84}
\begin{align} \label{86}
\begin{split}
\be^{1}&= \widetilde{E}^{1}\sqrt{\volbarg} - u\overline{e}_{2}(f)\overline{B}_{3}\\
&=\frac{\widetilde{E}^{1}}{\sqrt{\volg}} - \sqrt{\volg}\text{ }u\overline{e}_{2}(f)\left(\widetilde{B}_{3}
+v\times\widetilde{E}(e_{3})\right)\\
&=\frac{E^{1}}{\sqrt{\volg}} - \sqrt{\volg}  \left( v \times \left( \widetilde{B}+ v \times \widetilde{E})\right)\right)^{1}.
\end{split}
\end{align}
Analogous considerations lead to the remaining formulas of \eqref{65}.
\end{proof}

\begin{remark}\label{rem1}
Lemmas \ref{lemma1} and \ref{lemma2} imply that
\begin{equation}\label{87}
\overline{q}_{e} = \widetilde{q}_{e}=q_{e}, \text{ }\text{ }\text{ }\text{ }\text{ }\text{ }\text{ }\text{ } \overline{q}_{b}=\widetilde{q}_{b}=q_{b},
\end{equation}
when $(\byp,f)$ satisfy the asymptotic conditions described in Section \ref{sec2}, and $u$ remains bounded. Moreover Lemma \ref{lemma2}, \eqref{71}, \eqref{79}, and the anti-symmetry of cross products yield \eqref{37}.
\end{remark}

Together \eqref{47} and \eqref{65} show that the deformed electromagnetic field vanishes in the $\eta$-direction \eqref{43}. This condition is instrumental in establishing \eqref{36}, which is needed to impart the desired positivity property to the scalar curvature.

\begin{lemma}\label{lemma3}
If \eqref{43} holds, then
\begin{equation}\label{88}
|\overline{E}|_{\overline{g}}^{2}+|\overline{B}|_{\overline{g}}^{2}
=|\widetilde{E}|_{g}^{2}+|\widetilde{B}|_{g}^{2}+2\widetilde{E}\times \widetilde{B}(v).
\end{equation}
\end{lemma}

\begin{proof}
Using Lemma \ref{lemma2} and the notation contained in its proof, we find that
\begin{align} \label{89}
\begin{split}
|\widetilde{E}|_{g}^{2} =& \widetilde{E}^{i}\widetilde{E}^{j}g_{ij}     \\
=& \widetilde{E}^{i}\widetilde{E}^{j}\left(\delta_{ij}- \varphi \overline{e}_{i}(f) \overline{e}_{j}(f) - \by(\overline{e}_{i}) \overline{e}_{j}(f) - \by(\overline{e}_{j}) \overline{e}_{i}(f) \right) \\
=& (\widetilde{E}^{i})^{2}- \varphi (\widetilde{E}^{i}\overline{e}_{i}(f))^{2} - 2 \widetilde{E}^{3} \by(\overline{e}_{3})(\widetilde{E}^{i}\overline{e}_{i}(f))\\
=& \frac{1}{\volbarg}\left(\left( \be^{1} + u \overline{e}_{2}(f) \bm_{3} \right)^{2}
	+ \left( \be^{2} - u \overline{e}_{1}(f)\bm_{3}\right)^{2}\right) \\
&+ \frac{1}{\volbarg}\left(\be^{3} + u (\overline{e}_{1}(f)\bm_{2} - \overline{e}_{2}(f) \bm_{1}) + \by(\overline{e}_{3})\be(\barna f)\right)^{2}\\
&- \frac{\varphi\be(\barna f)^{2}}{\volbarg}
- \frac{2 \by(\overline{e}_{3}) \be(\barna f)}{\volbarg}\left(\be^{3} + u(\overline{e}_{1}(f)\bm_{2} - 				 \overline{e}_{2}(f) \bm_{1})+\by(\overline{e}_{3})\be(\barna f)\right)\\
=& \frac{1}{\volbarg}\left(\left( \be^{1} + u \overline{e}_{2}(f) \bm_{3} \right)^{2}
	+ \left(\be^{2} - u \overline{e}_{1}(f)\bm_{3} \right)^{2}\right) \\
	&+ \frac{1}{\volbarg}\left(\be^{3}
+ u(\overline{e}_{1}(f)\bm_{2} - \overline{e}_{2}(f) \bm_{1})\right)^{2}
- \frac{u^{2} \be(\barna f)^{2} }{\volbarg}. \\
\end{split}
\end{align}
Similarly
\begin{align} \label{90}
\begin{split}
|\widetilde{B}|_{g}^{2}
=& \frac{1}{\volbarg}\left(\left( \bm^{1} - u \overline{e}_{2}(f) \be_{3}\right)^{2}
	+ \left( \bm^{2} + u \overline{e}_{1}(f)\be_{3} \right)^{2}\right) \\
&+\frac{1}{\volbarg}\left(\bm^{3} + u(\overline{e}_{2}(f)\be_{1} - \overline{e}_{1}(f) \be_{2})  \right)^{2}- \frac{u^{2}\bm(\barna f)^{2} }{\volbarg}. \\
\end{split}
\end{align}
Next observe the identity
\begin{equation}\label{91}
v = \frac{u \nabla f}{\sqrt{\volg}} = \frac{u \barna f + u |\barna f|_{\bg}^{2} \by}{\sqrt{\volbarg}},
\end{equation}
so that
\begin{align} \label{92}
\begin{split}
&\widetilde{E} \times \widetilde{B} (v) \\
=& \epsilon_{ijl}\widetilde{E}^{i}\widetilde{B}^{j}v^{l} \\
=& \sqrt{\volbarg} \left(v^{1} \left(\widetilde{E}^{2}\widetilde{B}^{3}
-\widetilde{E}^{3}\widetilde{B}^{2}\right)
	+ v^{2} \left(\widetilde{E}^{3}\widetilde{B}^{1}-\widetilde{E}^{1}\widetilde{B}^{3}\right)
	+ v^{3} \left(\widetilde{E}^{1}\widetilde{B}^{2}-\widetilde{E}^{2}\widetilde{B}^{1}\right)\right) \\
=& \sqrt{\volbarg} \left(v^{3} \left(\widetilde{E}^{1}\widetilde{B}^{2}-\widetilde{E}^{2}\widetilde{B}^{1}\right)
	+ \widetilde{E}^{3} \left( v^{2}\widetilde{B}^{1} - v^{1}\widetilde{B}^{2}\right)
	+  \widetilde{B}^{3} \left( v^{1}\widetilde{E}^{2} - v^{2}\widetilde{E}^{1}\right)\right)\\
=& \frac{u |\barna f|^{2} \by(\overline{e}_{3})}{\volbarg}\left(\be^{1}\bm^{2}-\bm^{1}\be^{2}
	 + u \be(\barna f)\be_{3} + u \bm(\barna f) \bm_{3}\right) \\
&- \frac{\left( \be^{3}+ \by(\overline{e}_{3}) \be(\barna f)
    + u(\overline{e}_{1}(f)\bm_{2}-\overline{e}_{2}(f)\bm_{1})\right)}{\volbarg}
	\left(u(\overline{e}_{1}(f)\bm_{2}-\overline{e}_{2}(f)\bm_{1})
+ u^{2} |\barna f|_{\bg}^{2} \be_{3}\right) \\
&- \frac{\left(\bm^{3}+ \by(\overline{e}_{3}) \bm(\barna f)
	+ u(\overline{e}_{2}(f)\be_{1}-\overline{e}_{1}(f)\be_{2})\right)}{\volbarg}
	\left(u(\overline{e}_{2}(f)\be_{1}-\overline{e}_{1}(f)\be_{2})
+ u^{2} |\barna f|_{\bg}^{2}\bm_{3}\right) \\
=& - \frac{\left(\be^{3} +  u(\overline{e}_{1}(f)\bm_{2}
-\overline{e}_{2}(f)\bm_{1})\right)^{2}}{\volbarg}
	- \frac{\left(\bm^{3} + u(\overline{e}_{2}(f)\be_{1}
-\overline{e}_{1}(f)\be_{2})\right)^{2}}{\volbarg} \\
&+ \be^{3}\left(\be^{3} +  u(\overline{e}_{1}(f)\bm_{2}-\overline{e}_{2}(f)\bm_{1})\right)
	+ \bm^{3}\left(\bm^{3} + u(\overline{e}_{2}(f)\be_{1}-\overline{e}_{1}(f)\be_{2})\right).
\end{split}
\end{align}
Then combining \eqref{89}, \eqref{90}, and \eqref{92} produces
\begin{align} \label{93}
\begin{split}
& |\widetilde{E}|_{g}^{2} +  |\widetilde{B}|_{g}^{2} + 2 \widetilde{E} \times \widetilde{B}(v)  \\
=&\frac{1}{\volbarg}\left(\left( \be^{1} + u \overline{e}_{2}(f) \bm_{3} \right)^{2}
	+ \left( \be^{2} - u \overline{e}_{1}(f)\bm_{3} \right)^{2}\right) \\
& +\frac{1}{\volbarg}
\left(\left( \bm^{1} - u \overline{e}_{2}(f) \be_{3} \right)^{2}
	+ \left( \bm^{2} + u \overline{e}_{1}(f)\be_{3} \right)^{2}\right) \\
&- \frac{u^{2} \be(\barna f)^{2} }{\volbarg}
 	- \frac{u^{2}\bm(\barna f)^{2} }{\volbarg} \\
& - \frac{\left(\be^{3} +u(\overline{e}_{1}(f)\bm_{2}-\overline{e}_{2}(f)\bm_{1}) \right)^{2}}{\volbarg}
	- \frac{\left(\bm^{3}
+ u(\overline{e}_{2}(f)\be_{1}-\overline{e}_{1}(f)\be_{2}) \right)^{2}}{\volbarg}\\
& + 2\be^{3}\left(\be^{3} +  u(\overline{e}_{1}(f)\bm_{2}-\overline{e}_{2}(f)\bm_{1})\right)
	+ 2\bm^{3}\left(\bm^{3} + u(\overline{e}_{2}(f)\be_{1}-\overline{e}_{1}(f)\be_{2})\right)\\
=&\frac{1}{\volbarg}\left((\be^{1})^{2}+(\be^{2})^{2} +(\bm^{1})^{2}+(\bm^{2})^{2}\right)\\
&-\frac{u^{2}}{\volbarg}\left( \be(\barna f)^{2} + \bm(\barna f)^{2}
	 +(\overline{e}_{1}(f)\bm_{2}-\overline{e}_{2}(f)\bm_{1})^{2}
	+ (\overline{e}_{2}(f)\be_{1}-\overline{e}_{1}(f)\be_{2})^{2}\right) \\
&+ \be^{3}\left(\be^{3} +  2u(\overline{e}_{1}(f)\bm_{2}-\overline{e}_{2}(f)\bm_{1})\right)
	+ \bm^{3}\left(\bm^{3} + 2u(\overline{e}_{2}(f)\be_{1}-\overline{e}_{1}(f)\be_{2})\right) \\
=& (\be^{1})^{2}+(\be^{2})^{2} +(\bm^{1})^{2}+(\bm^{2})^{2} \\
&+ \be^{3}\left(\be^{3} +  2u(\overline{e}_{1}(f)\bm_{2}-\overline{e}_{2}(f)\bm_{1})\right)
	+ \bm^{3}\left(\bm^{3} + 2u(\overline{e}_{2}(f)\be_{1}-\overline{e}_{1}(f)\be_{2})\right).
\end{split}
\end{align}
The desired conclusion now follows from the assumption $\be(\eta)=\bm(\eta)=0$.
\end{proof}

We have now recorded enough properties of the deformed Maxwell field to obtain the final form of the scalar curvature identity.

\begin{cor}\label{cor1}
Under the hypotheses of Theorem \ref{thm1}, together with \eqref{38}, the scalar curvature of $\overline{g}$ is given by
\begin{align}\label{94}
\begin{split}
\overline{R}-|\overline{k}|_{\overline{g}}^{2}-2(|\overline{E}|_{\overline{g}}^{2}+|\overline{B}|_{\overline{g}}^{2})
=& 16\pi(\mu_{EM}-J_{EM}(v))+|k-\pi|_{g}^{2}+2u^{-1}div_{\overline{g}}(uQ)\\
&+2\left(E(e_{3})- \widetilde{E}(e_{3})\right)^{2}  + 2\left(B(e_{3})- \widetilde{B}(e_{3}) \right)^{2}.
\end{split}
\end{align}
\end{cor}

\begin{proof}
This follows directly from \eqref{39}, Lemmas \ref{lemma1} and \ref{lemma3}, and Remark \ref{rem1}.
\end{proof}

Lastly, we show that the deformed Maxwell field preserves the divergence free condition.

\begin{lemma} \label{lemma4}
If \eqref{43} holds, then
\begin{equation}\label{95}
div_{\bg}\be = \frac{div_{g} \widetilde{E}}{\sqrt{\volg}}
= \frac{div_{g} E}{\sqrt{\volg}}, \text{ }\text{ }\text{ }\text{ }\text{ }\text{ }
div_{\bg}\bm = \frac{div_{g} \widetilde{B}}{\sqrt{\volg}}= \frac{div_{g} B}{\sqrt{\volg}}.
\end{equation}
\end{lemma}

\begin{proof}
Note that the equalities relating the divergences of $(\widetilde{E},\widetilde{B})$ and $(E,B)$ were already given in Lemma \ref{lemma1}. Here we will compute the divergences of $(\overline{E},\overline{B})$ in terms of the divergences of $(\widetilde{E},\widetilde{B})$. Two different proofs will be given.
The first proof, presented below, is based on conceptual understanding, whereas the second proof is given in Appendix A and is based on direct computations.  In what follows $i,j,l\in\{1,2,3\}$, $a,b\in\{0,1,2,3\}$ and repeated indices will be summed.

Define the $4$-current as follows
\begin{equation} \label{96}
\overline{J}_{b}= \tilna^{a} \overline{F}_{ab},
\end{equation}
where $\tilna$ is the Levi-Civita connection for the metric $\tg= \bg-Y_{i}dx^{i}dt-\varphi dt^{2}$ on the spacetime $\overline{M}\times\mathbb{R}$. The relation between the $4$-current and the electric fields is given by
\begin{equation} \label{97}
div_{\bg} \be= \overline{J}(n), \quad \quad div_{g} \widetilde{E} = \overline{J}(N).
\end{equation}
Therefore
\begin{equation} \label{98}
div_{\bg}\be = \overline{J}(n)= \overline{J}\left(\sqrt{\volbarg}\text{ }N\right) - \overline{J}(u\barna f) = \sqrt{\volbarg}\text{ }div_{g} \widetilde{E} - \overline{J}(u\barna f).
\end{equation}

We will show that $\overline{J}(u \barna f)=0$ whenever $\be(\eta)=\bm(\eta)=0$, by directly computing $\overline{J}_{i}$. Observe that
%Recall that the Christoffel symbols for $(\mathbb{R} \times M, \tg)$ are as following.
%\begin{equation} \label{00014}
%\begin{split}
%& \tg^{tt}= -\frac{1}{u^{2}}, \qquad \tg^{ti}= - \frac{\by^{i}}{u^{2}},
%\qquad \tg^{ij}=\bg^{ij}-\frac{\by^{i}\by^{j}}{u^{2}} \\
%&\tilG^{t}_{tt}=0,  \qquad  \tilG^{k}_{tt}= \frac{1}{2} \phi^{k}, \qquad \tilG^{t}_{ij}=\frac{\bk_{ij}}{u}   \\
%&\tilG^{t}_{it}= \frac{\phi_{i}+\by^{l}(Y_{l\bar{;}i}-Y_{i\bar{;}l})}{2u^{2}}  \\
%&\tilG^{k}_{it}= -\frac{1}{2}\bg^{kl}(Y_{l\bar{;}i}-Y_{i\bar{;}l}) + \frac{\by^{k}}{2u^{2}}(\phi_{i}+\by^{l}%(Y_{l\bar{;}i}-Y_{i\bar{;}l}))   \\
%&\tilG^{k}_{ij}= \barG^{k}_{ij} + \frac{\by^{k}}{u}\bk_{ij}
%\end{split}
%\end{equation}
%We will compute the spatial component of $\overline{J}_{i}$ and show that $\overline{J}(u \barna f)=0$ with any %differentiable function $f$ which is axially symmetric, i.e. $f_{,\phi}=0$, for $()$.
\begin{align} \label{99}
\begin{split}
\overline{J}(\overline{e}_{i})=& \tilna^{a}\overline{F}_{ai}   \\
=& \overline{e}_{j}(\overline{F}_{ji})
+ \overline{F}(\tilna_{n} n, \overline{e}_{i})
+ \overline{F}(n, \tilna_{n}\overline{e}_{i})
- \overline{F}(\tilna_{\overline{e}_{j}}\overline{e}_{j}, \overline{e}_{i})
- \overline{F}(\overline{e}_{j}, \tilna_{\overline{e}_{j}}\overline{e}_{i})
 \\
=& \overline{e}_{j}(\overline{F}_{ji})
	- \bg(\barna_{\overline{e}_{j}}\overline{e}_{j}, \overline{e}_{l}) \overline{F}_{li}
	- \bg(\barna_{\overline{e}_{j}}\overline{e}_{i}, \overline{e}_{l}) \overline{F}_{jl}\\
&- \tg(\tilna_{\overline{e}_{j}}\overline{e}_{j}, n) \overline{F}_{0i}
	- \tg(\tilna_{\overline{e}_{j}}\overline{e}_{i}, n) \overline{F}_{j0}
	+ \tg(\tilna_{n}n, \overline{e}_{j}) \overline{F}_{ji}
	+ \tg(\tilna_{n}\overline{e}_{i}, \overline{e}_{j}) \overline{F}_{0j}.
\end{split}
\end{align}
Recall that $Tr_{\bg}\bk=0$, and that \eqref{29.1} the only nonzero components of $\bk$ are $\bk(\overline{e}_{3}, \overline{e}_{i})$, $i=1,2$. Thus
\begin{equation} \label{100}
\tg(\tilna_{\overline{e}_{j}}\overline{e}_{j}, n)
= - Tr_{\bg} \bk =0,
\end{equation}
\begin{equation} \label{101}
\tg(\tilna_{\overline{e}_{j}}\overline{e}_{i}, n)
= - \bk(\overline{e}_{j}, \overline{e}_{i})
= - \bk(\overline{e}_{3}, \overline{e}_{i})\delta_{j3},
\end{equation}
\begin{align} \label{102}
\begin{split}
 \tg(\tilna_{n}n, \overline{e}_{j})
&= -\tg(\tilna_{n}\overline{e}_{j}, n) \\
&= - \frac{1}{u} \tg(\tilna_{\partial_{t} + \by} \overline{e}_{j}, n) \\
&= - \frac{1}{u} \left(\tg(\tilna_{\overline{e}_{j}}(un - \by), n)
-\bk(\overline{e}_{j}, \by) \right) \\
&= u^{-1}\overline{e}_{j}(u),
\end{split}
\end{align}
and
\begin{align} \label{103}
\begin{split}
\tg(\tilna_{n}\overline{e}_{i}, \overline{e}_{j})
&= \frac{1}{u} \tg(\tilna_{\partial_{t} + \by} \overline{e}_{i}, \overline{e}_{j}) \\
&= \frac{1}{u} \tg(\tilna_{\overline{e}_{i}}(u n - \by), \overline{e}_{j})
+ \frac{1}{u}\tg(\tilna_{\by} \overline{e}_{i}, \overline{e}_{j}) \\
&= \bk(\overline{e_{i}}, \overline{e}_{3})\delta_{3j}
- u^{-1}|\eta| \overline{e}_{i}(\byp) \delta_{3j},
\end{split}
\end{align}
where in the last two equations the fact that $\partial_{t}$ is a Killing field was utilized.
Substituting these expressions into \eqref{99} produces
\begin{equation} \label{104}
\overline{J}(\overline{e}_{i})
= \overline{e}_{j}(\overline{F}_{ji})
	- \bg(\barna_{\overline{e}_{j}}\overline{e}_{j}, \overline{e}_{l}) \overline{F}_{li}
	- \bg(\barna_{\overline{e}_{j}}\overline{e}_{i}, \overline{e}_{l}) \overline{F}_{jl}
	+ u^{-1}\overline{e}_{j}(u) \overline{F}_{ji}
	- u^{-1}|\eta| \overline{e}_{i}(\byp)\overline{F}_{03}.
\end{equation}

Note that $\overline{e}_{3}=|\eta|^{-1}\eta$, and so $\bg(\barna_{ \overline{e}_{j}} \overline{e}_{j}, \overline{e}_{3})=0$ by the Killing equation. Hence
\begin{align} \label{105}
\begin{split}
\overline{J}(\overline{e}_{1})
=& \overline{e}_{2}(\overline{F}_{21})
	- \bg(\barna_{\overline{e}_{j}}\overline{e}_{j}, \overline{e}_{2}) \overline{F}_{21}
   - \left(\bg(\barna_{\overline{e}_{1}}\overline{e}_{1}, \overline{e}_{2}) \overline{F}_{12}
			+\bg(\barna_{\overline{e}_{3}}\overline{e}_{1}, \overline{e}_{2}) \overline{F}_{32}
			+\bg(\barna_{\overline{e}_{2}}\overline{e}_{1}, \overline{e}_{3}) \overline{F}_{23}\right)\\
&	+ u^{-1}\overline{e}_{2}(u) \overline{F}_{21}
	- u^{-1}|\eta| \overline{e}_{1}(\byp)\overline{F}_{03}\\
=& 	\overline{e}_{2}(\overline{F}_{21})
	- \bg(\barna_{\overline{e}_{j}}\overline{e}_{j}, \overline{e}_{2}) \overline{F}_{21}
    - \bg(\barna_{\overline{e}_{1}}\overline{e}_{1}, \overline{e}_{2}) \overline{F}_{12}
	+ u^{-1}\overline{e}_{2}(u) \overline{F}_{21}
	- u^{-1}|\eta| \overline{e}_{1}(\byp)\overline{F}_{03},
\end{split}
\end{align}
and
\begin{equation} \label{106}
\overline{J}(\overline{e}_{2})
= \overline{e}_{1}(\overline{F}_{12})
	- \bg(\barna_{\overline{e}_{j}}\overline{e}_{j}, \overline{e}_{1}) \overline{F}_{12}
	- \bg(\barna_{\overline{e}_{2}}\overline{e}_{2}, \overline{e}_{1}) \overline{F}_{21}
	+ u^{-1}\overline{e}_{1}(u) \overline{F}_{12}
	- u^{-1}|\eta| \overline{e}_{2}(\byp)\overline{F}_{03}.
\end{equation}
Therefore $\overline{J}(\overline{e}_{1})=\overline{J}(\overline{e}_{2})=0$ since $\be_{3}=\bm_{3}=0$, and in particular
\begin{equation}\label{107}
\overline{J}(\barna f)
= \overline{e}_{1}(f)\overline{J}(\overline{e}_{1})+\overline{e}_{2}(f)\overline{J}(\overline{e}_{2}) =0.
\end{equation}
This gives the desired result for the divergence of the electric fields.

An analogous procedure with $\overline{F}$ replaced by $\ast \overline{F}$, where $\ast$ is the Hodge star operator with respect to the metric $\widetilde{g}$, yields the desired result for the divergence of the magnetic fields.
\end{proof}

\section{The Reduction Argument and Case of Equality}
\label{sec4} \setcounter{equation}{0}
\setcounter{section}{4}

We will now follow the maximal case proof of the mass-angular momentum-charge inequality, within the setting of the deformed initial data $(\overline{M},\overline{g},\overline{k},\overline{E},\overline{B})$.
The main difficulty is a lack of the pointwise scalar curvature inequality as appearing in \eqref{21}. However a judicious choice of $u$ will overcome this problem. Before explaining this further, it is necessary to introduce the appropriate potentials. In electrovacuum the existence of potentials is well-known, and the proof relies on the full Maxwell equations, see for example \cite{Costa}, \cite{Weinstein}. In the current setting the initial data are not necessarily electrovacuum, and we do not have knowledge of the full Maxwell equations, but rather just Gauss's Law. Nevertheless, the desired potentials still exist under our assumptions.

\begin{lemma} \label{lemma5}
Let $\{\overline{\theta}^{0},\overline{\theta}^{i}\}$ be the dual co-frame to $\{n,\overline{e}_{i}\}$, and assume that $\overline{E}(\eta)=\overline{B}(\eta)=0$. Then
\begin{equation}\label{108}
d (\overline{F}(\eta, \cdot) )=|\eta|(div_{\overline{g}}\overline{B})
\overline{\theta}^{2}\wedge\overline{\theta}^{1},\text{ }\text{ }\text{ }\text{ }\text{ }
d(\ast \overline{F}(\eta, \cdot))=|\eta|(div_{\overline{g}}\overline{E})
\overline{\theta}^{2}\wedge\overline{\theta}^{1}.
\end{equation}
In particular, if $(\overline{E}, \overline{B})$ are divergence free, then there exist potentials for the electromagnetic field such that
\begin{equation}\label{109}
d\overline{\psi}=\overline{F}(\eta, \cdot),\text{ }\text{ }\text{ }\text{ }\text{ }d\overline{\chi}
=\ast \overline{F} (\eta, \cdot).
\end{equation}
Moreover
\begin{equation}\label{110}
d\left(\overline{k}(\eta)\times\eta-\overline{\chi}d\overline{\psi}
+\overline{\psi}d\overline{\chi}\right)=|\eta|\left(\overline{J}_{EM}(\eta)-\overline{\chi}div_{\overline{g}}\overline{B}
+\overline{\psi}div_{\overline{g}}\overline{E}\right)\overline{\theta}^{2}\wedge\overline{\theta}^{1},
\end{equation}
so that if in addition $\overline{J}_{EM}(\eta)=0$, then there exists a charged twist potential
\begin{equation}\label{111}
d\overline{\omega}=\overline{k}(\eta)\times\eta-\overline{\chi}d\overline{\psi}
+\overline{\psi}d\overline{\chi}.
\end{equation}
\end{lemma}

\begin{proof}
In the given frame, the field strength and its Hodge dual take the form
\begin{align} \label{112}
\begin{split}
\overline{F}(\eta, \cdot)
&= |\eta | \left(\be(\overline{e}_{3})\overline{\theta}^{0} + \bm(\overline{e}_{2})\overline{\theta}^{1}
-  \bm(\overline{e}_{1})\overline{\theta}^{2}\right)=
 |\eta | \left(\bm(\overline{e}_{2})\overline{\theta}^{1}
-  \bm(\overline{e}_{1})\overline{\theta}^{2}\right),\\
\ast \overline{F}(\eta, \cdot)
&= |\eta | \left(\bm(\overline{e}_{3})\overline{\theta}^{0} + \be(\overline{e}_{2})\overline{\theta}^{1}
- \be(\overline{e}_{1})\overline{\theta}^{2}\right)
=|\eta | \left(\be(\overline{e}_{2})\overline{\theta}^{1}
- \be(\overline{e}_{1})\overline{\theta}^{2}\right),
\end{split}
\end{align}
since $\overline{E}(\eta)=\overline{B}(\eta)=0$. A basic calculation shows that
\begin{equation}\label{113}
d\overline{\theta}^{1}=\overline{e}_{2}\left(\log e^{-\overline{U}+\overline{\alpha}}\right)\overline{\theta}^{2}\wedge\overline{\theta}^{1},\text{ }\text{ }\text{ }\text{ }\text{ }d\overline{\theta}^{2}=\overline{e}_{1}\left(\log e^{-\overline{U}+\overline{\alpha}}\right)\overline{\theta}^{1}\wedge\overline{\theta}^{2},
\end{equation}
\begin{equation}\label{114}
d\overline{\theta}^{3}=\overline{e}_{1}\left(\log \overline{\rho}e^{-\overline{U}}\right)\overline{\theta}^{1}\wedge\overline{\theta}^{3}
+\overline{e}_{2}\left(\log \overline{\rho}e^{-\overline{U}}\right)\overline{\theta}^{2}\wedge\overline{\theta}^{3}
+ \overline{\rho}e^{\overline{U}-2\overline{\alpha}}(A_{\overline{\rho},\overline{z}}
-A_{\overline{z},\overline{\rho}})\overline{\theta}^{2}\wedge\overline{\theta}^{1}.
\end{equation}
Then
\begin{equation}\label{115}
d\left(\ast\overline{F}(\eta,\cdot)\right)
=\sum_{i=1}^{2}\left(\overline{e}_{i}\left(|\eta|\overline{E}(\overline{e}_{i})\right)
+\overline{e}_{i}\left(\log e^{-\overline{U}+\overline{\alpha}}\right)|\eta|\overline{E}(\overline{e}_{i})\right)
\overline{\theta}^{2}\wedge\overline{\theta}^{1}.
\end{equation}
Since $|\eta|=\overline{\rho}e^{-\overline{U}}$ and
\begin{align}\label{116}
\begin{split}
div_{\overline{g}}\overline{E}&=\sum_{i=1}^{2}\overline{e}_{i}\left(\overline{E}(\overline{e}_{i})\right)
-\sum_{i=1}^{3}\sum_{j=1}^{2}\overline{g}(\overline{\nabla}_{\overline{e}_{i}}\overline{e}_{i},\overline{e}_{j})
\overline{E}(\overline{e}_{j})\\
&=\sum_{i=1}^{2}\overline{e}_{i}\left(\overline{E}(\overline{e}_{i})\right)
+\sum_{j=1}^{2}\overline{e}_{j}\left(\log \overline{\rho}e^{-2\overline{U}+\overline{\alpha}}\right)
\overline{E}(\overline{e}_{j}),
\end{split}
\end{align}
the second equation in \eqref{108} follows. Similar computations yield the first equation in \eqref{108}.

Next recall \cite{DainKhuriWeinsteinYamada} that
\begin{equation}\label{117}
d(\overline{k}(\eta)\times\eta)
=|\eta|div_{\overline{g}}\overline{k}(\eta)\overline{\theta}^{2}\wedge\overline{\theta}^{1}.
\end{equation}
Furthermore since
\begin{equation}\label{118}
\overline{e}_{1}(\overline{\psi})=|\eta|\overline{B}(\overline{e}_{2}),\text{ }\text{ }\text{ }\text{ }\text{ }\overline{e}_{2}(\overline{\psi})=-|\eta|\overline{B}(\overline{e}_{1}),\text{ }\text{ }\text{ }\text{ }\text{ }\overline{e}_{1}(\overline{\chi})=|\eta|\overline{E}(\overline{e}_{2}),\text{ }\text{ }\text{ }\text{ }\text{ }\overline{e}_{2}(\overline{\chi})=-|\eta|\overline{E}(\overline{e}_{1}),
\end{equation}
we have
\begin{align}\label{119}
\begin{split}
&d\left([-\overline{\chi}\overline{e}_{1}(\overline{\psi})+\overline{\psi}\overline{e}_{1}(\overline{\chi})]
\overline{\theta}^{1}
+[-\overline{\chi}\overline{e}_{2}(\overline{\psi})+\overline{\psi}\overline{e}_{2}(\overline{\chi})]
\overline{\theta}^{2}\right)\\
=&\overline{e}_{2}[-\overline{\chi}\overline{e}_{1}(\overline{\psi})+\overline{\psi}\overline{e}_{1}(\overline{\chi})]
\overline{\theta}^{2}\wedge\overline{\theta}^{1}
+[-\overline{\chi}\overline{e}_{1}(\overline{\psi})+\overline{\psi}\overline{e}_{1}(\overline{\chi})]
\overline{e}_{2}\left(\log e^{-\overline{U}+\overline{\alpha}}\right)\overline{\theta}^{2}\wedge\overline{\theta}^{1}\\
&+\overline{e}_{1}[-\overline{\chi}\overline{e}_{2}(\overline{\psi})+\overline{\psi}\overline{e}_{2}(\overline{\chi})]
\overline{\theta}^{1}\wedge\overline{\theta}^{2}
+[-\overline{\chi}\overline{e}_{2}(\overline{\psi})+\overline{\psi}\overline{e}_{2}(\overline{\chi})]
\overline{e}_{1}\left(\log e^{-\overline{U}+\overline{\alpha}}\right)\overline{\theta}^{1}\wedge\overline{\theta}^{2}\\
=&2|\eta|^{2}\left(\overline{E}(\overline{e}_{1})\overline{B}(\overline{e}_{2})
-\overline{E}(\overline{e}_{2})\overline{B}(\overline{e}_{1})\right)\overline{\theta}^{2}\wedge\overline{\theta}^{1}\\
&-|\eta|\overline{\chi}\sum_{i=1}^{2}\left[\overline{e}_{i}\left(|\eta|\overline{B}(\overline{e}_{i})\right)
+|\eta|\overline{B}(\overline{e}_{i})\overline{e}_{i}\left(\log e^{-\overline{U}+\overline{\alpha}}\right)\right]\overline{\theta}^{2}\wedge\overline{\theta}^{1}\\
&+|\eta|\overline{\psi}\sum_{i=1}^{2}\left[\overline{e}_{i}\left(|\eta|\overline{E}(\overline{e}_{i})\right)
+|\eta|\overline{E}(\overline{e}_{i})\overline{e}_{i}\left(\log e^{-\overline{U}+\overline{\alpha}}\right)\right]\overline{\theta}^{2}\wedge\overline{\theta}^{1}\\
=&|\eta|\left(2\overline{E}\times\overline{B}(\eta)-\overline{\chi}div_{\overline{g}}\overline{B}
+\overline{\psi}div_{\overline{g}}\overline{E}\right)\overline{\theta}^{2}\wedge\overline{\theta}^{1}.
\end{split}
\end{align}
Combining \eqref{117} and \eqref{119} yields \eqref{110}.

Lastly, the potentials exist since $\overline{M}$ is simply connected.
\end{proof}

\begin{remark}\label{remark2}
In equation \eqref{10}, a seemingly different equation was used to define a charged twist potential for the given initial data $(M,g,k,E,B)$. Although this equation appears to be different from \eqref{111}, when
both are applied to the data $(\overline{M},\overline{g},\overline{k},\overline{E},\overline{B})$, they coincide and realize the same result. This will be shown in Appendix B.
\end{remark}

We now show how to choose the remaining unknown in the description of the deformed initial data, namely $u$. Recall that in Brill coordinates, the mass may be expressed in a simple formula (\cite{Brill}, \cite{Dain0})
\begin{equation}\label{120}
\overline{m}
=\frac{1}{32\pi}\int_{\mathbb{R}^{3}}\left(2e^{-2\overline{U}+2\overline{\alpha}}\!\text{ }\overline{R}
+\overline{\rho}^{2}e^{-2\overline{\alpha}}
(A_{\overline{\rho},\overline{z}}-A_{\overline{z},\overline{\rho}})^{2}
+4|\partial\overline{U}|^{2}\right)dx,
\end{equation}
where $|\partial\overline{U}|$ and $dx$ denote the Euclidean norm and volume element. Let
\begin{equation}\label{121}
\mathcal{M}(\overline{U},\overline{\omega},\overline{\chi},\overline{\psi})
=\frac{1}{8\pi}\int_{\mathbb{R}^{3}}\left(|\partial\overline{U}|^{2}
+\frac{e^{4\overline{U}}}{\overline{\rho}^{4}}|\partial\overline{\omega}
+\overline{\chi}\partial\overline{\psi}-\overline{\psi}\partial\overline{\chi}|^{2}
+\frac{e^{2\overline{U}}}{\overline{\rho}^{2}}\left(|\partial\overline{\chi}|^{2}
+|\partial\overline{\psi}|^{2}\right)\right)dx,
\end{equation}
then the mass formula may be reexpressed as
\begin{align}\label{122}
\begin{split}
\overline{m}-\mathcal{M}(\overline{U},\overline{\omega},\overline{\chi},\overline{\psi})
=&\frac{1}{32\pi}\int_{\mathbb{R}^{3}}\left(2e^{-2\overline{U}+2\overline{\alpha}}\!\text{ }\overline{R}
+\overline{\rho}^{2}e^{-2\overline{\alpha}}
(A_{\overline{\rho},\overline{z}}-A_{\overline{z},\overline{\rho}})^{2}\right)dx\\
&-\frac{1}{8\pi}\int_{\mathbb{R}^{3}}\left(\frac{e^{4\overline{U}}}{\overline{\rho}^{4}}|\partial\overline{\omega}
+\overline{\chi}\partial\overline{\psi}-\overline{\psi}\partial\overline{\chi}|^{2}
+\frac{e^{2\overline{U}}}{\overline{\rho}^{2}}\left(|\partial\overline{\chi}|^{2}
+|\partial\overline{\psi}|^{2}\right)\right)dx.
\end{split}
\end{align}
The goal is to show that the right-hand side is nonnegative, by using the lower bound for scalar curvature to dominate the terms on the second line. In this regard, we record the following observation.

\begin{lemma}\label{lemma6}
Under the assumptions of Lemma \ref{lemma5} guaranteeing the existence of potentials
\begin{equation}\label{123}
|\overline{E}|_{\overline{g}}^{2}+|\overline{B}|_{\overline{g}}^{2}
=\frac{e^{4\overline{U}-2\overline{\alpha}}}{\overline{\rho}^{2}}\left(|\partial\overline{\chi}|^{2}
+|\partial\overline{\psi}|^{2}\right),\text{ }\text{ }\text{ }\text{ }\text{ }
|\overline{k}|_{\overline{g}}^{2}
=2\frac{e^{6\overline{U}-2\overline{\alpha}}}{\overline{\rho}^{4}}|\partial\overline{\omega}
+\overline{\chi}\partial\overline{\psi}-\overline{\psi}\partial\overline{\chi}|^{2}.
\end{equation}
\end{lemma}

\begin{proof}
A direct computation using \eqref{118} produces
\begin{align}\label{124}
\begin{split}
|\overline{E}|_{\overline{g}}^{2}+|\overline{B}|_{\overline{g}}^{2}&=
\overline{E}(\overline{e}_{1})^{2}+\overline{E}(\overline{e}_{2})^{2}
+\overline{B}(\overline{e}_{1})^{2}+\overline{B}(\overline{e}_{2})^{2}\\
&=\frac{e^{2\overline{U}}}{\overline{\rho}^{2}}\left(\overline{e}_{1}(\overline{\chi})^{2}+
\overline{e}_{2}(\overline{\chi})^{2}+\overline{e}_{1}(\overline{\psi})^{2}
+\overline{e}_{2}(\overline{\psi})^{2}\right)\\
&=\frac{e^{4\overline{U}-2\overline{\alpha}}}{\overline{\rho}^{2}}
\left[(\partial_{\overline{\rho}}\overline{\chi})^{2}+
(\partial_{\overline{z}}\overline{\chi})^{2}+(\partial_{\overline{\rho}}\overline{\psi})^{2}+
(\partial_{\overline{z}}\overline{\psi})^{2}\right]\\
&=\frac{e^{4\overline{U}-2\overline{\alpha}}}{\overline{\rho}^{2}}
\left(|\partial\overline{\chi}|^{2}+|\partial\overline{\psi}|^{2}\right).
\end{split}
\end{align}
Moreover from \eqref{111} we have
\begin{equation}\label{125}
\overline{k}(\eta)\times\eta=d\overline{\omega}+\overline{\chi}d\overline{\psi}
-\overline{\psi}d\overline{\chi}.
\end{equation}
Hence
\begin{equation}\label{126}
\overline{k}(\overline{e}_{1},\overline{e}_{3})
=-|\eta|^{-2}\left(\overline{e}_{2}(\overline{\omega})
+\overline{\chi}\overline{e}_{2}(\overline{\psi})
-\overline{\psi}\overline{e}_{2}(\overline{\chi})\right)
=-|\eta|^{-2}e^{\overline{U}-\overline{\alpha}}
\left(\partial_{\overline{z}}\overline{\omega}
+\overline{\chi}\partial_{\overline{z}}\overline{\psi}
-\overline{\psi}\partial_{\overline{z}}\overline{\chi}\right)
\end{equation}
and
\begin{equation}\label{127}
\overline{k}(\overline{e}_{2},\overline{e}_{3})
=|\eta|^{-2}\left(\overline{e}_{1}(\overline{\omega})
+\overline{\chi}\overline{e}_{1}(\overline{\psi})
-\overline{\psi}\overline{e}_{1}(\overline{\chi})\right)
=|\eta|^{-2}e^{\overline{U}-\overline{\alpha}}
\left(\partial_{\overline{\rho}}\overline{\omega}
+\overline{\chi}\partial_{\overline{\rho}}\overline{\psi}
-\overline{\psi}\partial_{\overline{\rho}}\overline{\chi}\right).
\end{equation}
Now use \eqref{29.1} to find
\begin{equation}\label{128}
|\overline{k}|_{\overline{g}}^{2}=2\left(\overline{k}(\overline{e}_{1},\overline{e}_{3})^{2}
+\overline{k}(\overline{e}_{2},\overline{e}_{3})^{2}\right).
\end{equation}
The desired conclusion follows by combining \eqref{126}, \eqref{127}, and \eqref{128}.
\end{proof}

Continuing with the expression \eqref{122}, an application of Corollary \ref{cor1} and the charged dominant energy condition \eqref{2} produces
\begin{align}\label{129}
\begin{split}
\overline{m}-\mathcal{M}(\overline{U},\overline{\omega},\overline{\chi},\overline{\psi})
&\geq\frac{1}{32\pi}\int_{\mathbb{R}^{3}}2e^{-2\overline{U}+2\overline{\alpha}}\left[\overline{R}
-|\overline{k}|_{\overline{g}}^{2}-2(|\overline{E}|_{\overline{g}}^{2}+|\overline{B}|_{\overline{g}}^{2})\right]dx\\
&\geq\frac{1}{8\pi}\int_{\mathbb{R}^{3}}\frac{e^{-2\overline{U}+2\overline{\alpha}}}{u}div_{\overline{g}}(uQ)dx\\
&\geq\frac{1}{8\pi}\int_{\overline{M}}\frac{e^{\overline{U}}}{u}div_{\overline{g}}(uQ)dx_{\overline{g}},
\end{split}
\end{align}
where the volume element for $\overline{g}$ is given by $dx_{\overline{g}}=e^{-3\overline{U}+2\overline{\alpha}}dx$. This strongly suggests that we choose
\begin{equation}\label{130}
u=e^{\overline{U}}=\frac{\overline{\rho}}{\sqrt{\overline{g}_{\phi\phi}}}=\frac{\overline{\rho}}{\sqrt{g_{\phi\phi}}}.
\end{equation}
If $\overline{g}$ preserves the asymptotic geometry of $g$, then in light of \eqref{17}, \eqref{18}, \eqref{19}
\begin{equation}\label{131}
u=1+o_{l-3}(r^{-\frac{1}{2}})\text{ }\text{ }\text{ as }\text{ }\text{ }r\rightarrow\infty\text{ }\text{ }
\text{ in }\text{ }\text{ }M_{end}^{0},
\end{equation}
\begin{equation}\label{132}
u=r_{n}^{2}+o_{l-4}(r_{n}^{\frac{5}{2}})\text{ }\text{ }\text{ as }\text{ }\text{ }r_{n}\rightarrow 0\text{ }\text{ }\text{ in asymptotically flat}\text{ }\text{ }M_{end}^{n},
\end{equation}
\begin{equation}\label{133}
u=r_{n}+o_{l-4}(r_{n}^{\frac{3}{2}})\text{ }\text{ }\text{ as }\text{ }\text{ }r_{n}\rightarrow 0\text{ }\text{ }\text{ in asymptotically cylindrical}\text{ }\text{ }M_{end}^{n},
\end{equation}
where $r_{n}$ is the Euclidean distance to the point $i_{n}$ defining the end. Thus, with the aid of
the asymptotics for $f$ \eqref{41}, \eqref{42} and $Y^{\phi}$ \eqref{31}, as well as the following bounds which are implied by \eqref{3}, \eqref{19.1}, \eqref{19.2}
\begin{equation} \label{134}
|k|_{g}+|k(\partial_{\phi},\cdot)|_{g}+|k(\partial_{\phi},\partial_{\phi})|\leq c\text{ }\text{ }\text{ on }
\text{ }\text{ }M,
\end{equation}
the asymptotic boundary integrals arising from the right-hand side of \eqref{129} all vanish as long as $\mathcal{J}=\overline{\mathcal{J}}$ (see Appendix C in \cite{ChaKhuri}). Therefore
\begin{equation}\label{135}
\overline{m}\geq\mathcal{M}(\overline{U},\overline{\omega},\overline{\chi},\overline{\psi}).
\end{equation}

\begin{theorem}\label{thm2}
Let $(M,g,k,E,B)$ be a smooth, simply connected, axially symmetric initial data set satisfying the charged dominant energy condition \eqref{2} and condition \eqref{9}, and with two ends, one designated asymptotically flat and the other either asymptotically flat or asymptotically cylindrical. If the system of equations \eqref{30}, \eqref{38}, \eqref{130} admits a smooth solution $(u,Y^{\phi},f)$ satisfying the asymptotics \eqref{31}, \eqref{40}-\eqref{42}, \eqref{131}-\eqref{133}, then
\begin{equation}\label{136}
m^2\geq\frac{q^{2}+\sqrt{q^{4}+4\mathcal{J}^{2}}}{2},
\end{equation}
and equality holds if and only if $(M,g,k,E,B)$ arises from an embedding into the extreme Kerr-Newman spacetime.
\end{theorem}

\begin{proof}
The existence of a solution $(u,Y^{\phi},f)$ allows an application of the maximal case proof to the deformed
initial data $(\overline{M},\overline{g},\overline{k},\overline{E},\overline{B})$ as above, arriving at the inequality \eqref{135}. Previous results of \cite{ChruscielCosta}, \cite{Costa}, \cite{SchoenZhou} then show that
\begin{equation}\label{137}
\mathcal{M}(\overline{U},\overline{\omega},\overline{\chi},\overline{\psi})
\geq\sqrt{\frac{\overline{\mathcal{J}}^{2}}{\overline{m}^{2}}+\overline{q}^{2}}.
\end{equation}
Furthermore according to \eqref{21} and \eqref{22} $\overline{m}=m$, $\overline{\mathcal{J}}=\mathcal{J}$, and $\overline{q}=q$, so that \eqref{135} gives
\begin{equation}\label{138}
m\geq\sqrt{\frac{\mathcal{J}^{2}}{m^{2}}+q^{2}},
\end{equation}
which is equivalent to \eqref{136}.

Consider now the case of equality in \eqref{136}. Through the process of deriving \eqref{135}, several nonnegative terms were omitted from the right-hand side. These terms arise from Corollary \ref{cor1} and \eqref{122}. Under the current assumption, they must all vanish
\begin{equation}\label{139}
|\mu_{EM}-J_{EM}(v)|=|k-\pi|_{g}=|A_{\overline{\rho},\overline{z}}-A_{\overline{z},\overline{\rho}}|
=|E(e_{3})- \widetilde{E}(e_{3})|=|B(e_{3})- \widetilde{B}(e_{3})|=0.
\end{equation}
As a result of the charged dominant energy condition, the fact that $|v|_{g}<1$, and the identity
\begin{equation}\label{140}
\mu_{EM}-J_{EM}(v)=(\mu_{EM}-|J_{EM}|_{g})+(1-|v|_{g})|J_{EM}|_{g}+(|J_{EM}|_{g}|v|_{g}-J_{EM}(v)),
\end{equation}
we have
\begin{equation}\label{141}
\mu_{EM}=|J_{EM}|_{g}=0.
\end{equation}

It will now be shown that $(\overline{M},\overline{g},\overline{k},\overline{E},\overline{B})$ is an electrovacuum initial data set. Since
$Tr_{\overline{g}}\overline{k}=0$
\begin{equation}\label{142}
8\pi \overline{J}_{EM}=div_{\overline{g}}\overline{k}+2\overline{E}\times\overline{B}.
\end{equation}
From \eqref{30} we know that $\overline{J}_{EM}(\overline{e}_{3})=0$, and with the help of $|A_{\overline{\rho},\overline{z}}-A_{\overline{z},\overline{\rho}}|=0$, it is shown in \cite{ChaKhuri} that $div_{\overline{g}}\overline{k}(\overline{e}_{i})=0$, $i=1,2$. Moreover since $\overline{E}(\eta)=\overline{B}(\eta)=0$, it holds that $\overline{E}\times\overline{B}(\overline{e}_{i})=0$, $i=1,2$. Hence $|\overline{J}_{EM}|_{\overline{g}}=0$.

Consider now the energy density after the contribution from the electromagnetic field has been removed
\begin{equation}\label{143}
16\pi \overline{\mu}_{EM}=\overline{R}+(Tr_{\overline{g}}\overline{k})^{2}
-|\overline{k}|_{\overline{g}}^{2}-2(|\overline{E}|_{\overline{g}}^{2}
+|\overline{B}|_{\overline{g}}^{2})
=\overline{R}-|\overline{k}|_{\overline{g}}^{2}-2(|\overline{E}|_{\overline{g}}^{2}
+|\overline{B}|_{\overline{g}}^{2}).
\end{equation}
The computation (7.11) in Appendix A of \cite{ChaKhuri} yields
\begin{equation}\label{144}
\overline{R}-|\overline{k}|_{\overline{g}}^{2}
=-2(div_{\overline{g}}\overline{k})(u\overline{\nabla} f)
+16\pi(\mu-J(v))+|k|_{g}^{2}-|\pi|_{g}^{2}+2(div_{g}k)(v)-2(div_{g}\pi)(v)
\end{equation}
when equation \eqref{38} is satisfied, where $\mu$ and $J$ are the full matter and momentum density of the matter fields. However, $|div_{\overline{g}}\overline{k}|_{\overline{g}}=0$ and $|k-\pi|_{g}=0$ imply that
\begin{equation}\label{145}
\overline{R}-|\overline{k}|_{\overline{g}}^{2}
=16\pi(\mu-J(v))=2(|E|_{g}^{2}+|B|_{g}^{2})+4E\times B(v)=2(|\overline{E}|_{\overline{g}}^{2}
+|\overline{B}|_{\overline{g}}^{2}),
\end{equation}
where Lemma \ref{lemma3} together with $\widetilde{E}=E$ and $\widetilde{B}=B$ (which follows from the definition of $\widetilde{E}$, $\widetilde{B}$ and \eqref{139}) were used in the last equality. Therefore $\overline{\mu}_{EM}=0$, and $(\overline{M},\overline{g},\overline{k},\overline{E},\overline{B})$
is an electrovacuum initial data set.

Next, since the deformed initial data are electrovacuum and
\begin{equation}\label{146}
\overline{m}\geq\sqrt{\frac{\overline{\mathcal{J}}^{2}}{\overline{m}^{2}}+\overline{q}^{2}},
\end{equation}
the results of \cite{SchoenZhou} apply to show that $(\overline{M},\overline{g},\overline{k},\overline{E},\overline{B})$ is isometric to the data set on the $t=0$ slice
$(\mathbb{R}^{3}-\{0\},g_{EKN},k_{EKN},E_{EKN},B_{EKN})$ of the extreme Kerr-Newman spacetime $\mathbb{EKN}^{4}$. Consider the map $M\rightarrow\mathbb{EKN}^{4}$ given by $x\mapsto(x,f(x))$. The induced metric on the graph is given by
\begin{equation}\label{147}
(g_{EKN})_{ij}-f_{i}(Y_{EKN})_{j}-f_{j}(Y_{EKN})_{i}-(u_{EKN}^{2}-|Y_{EKN}|_{g_{EKN}}^{2}) f_{i}f_{j},
\end{equation}
where
\begin{equation}\label{148}
(k_{EKN})_{ij}=\frac{1}{2u_{EKN}}\left(\nabla^{EKN}_{i}(Y_{EKN})_{j}+\nabla^{EKN}_{j}(Y_{EKN})_{i}\right),
\end{equation}
and $(u_{EKN},-Y_{EKN})$ are the lapse and shift. If $\partial_{\phi}$ denotes the spacelike Killing field
in this spacetime, then $g_{EKN}^{ij}(Y_{EKN})_{j}\partial_{i}=Y_{EKN}^{\phi}\partial_{\phi}$, and $Y_{EKN}^{\phi}$
satisfies equation \eqref{30} with $(\overline{g},\overline{E},\overline{B})$ replaced by
$(g_{EKN},E_{EKN},B_{EKN})$, as well as boundary condition \eqref{31}. Since there is a unique solution to \eqref{30}, \eqref{31}, and $\overline{g}\cong g_{EKN}$, we have that $Y=Y_{EKN}$. Moreover a direct calculation shows that $u_{EKN}=e^{U_{EKN}}=e^{\overline{U}}=u$, where
$U_{EKN}$ arises from the Brill coordinate expression for $g_{EKN}$. It now follows from \eqref{25} and \eqref{26}
that $g$ agrees with the induced metric \eqref{147}. Note also that since \eqref{139} implies $\pi=k$, the
second fundamental form of the embedding $(M,g)\hookrightarrow\mathbb{EKN}^{4}$ is given by $k$. Furthermore, in light of the fact that $\widetilde{E}=E$, $\widetilde{B}=B$, the construction of $\overline{E}$, $\overline{B}$ guarantees that $E$, $B$ arise as the induced electric and magnetic field on the graph. Thus
the initial data $(M,g,k,E,B)$ arise from the extreme Kerr-Newman spacetime.

Lastly, if $(M,g,k,E,B)$ arises from extreme Kerr-Newman, then by the properties of this spacetime, equality in \eqref{136} holds.
\end{proof}

Theorem \ref{thm2} reduces the proof of the mass-angular momentum-charge inequality, in the general non-maximal case, to  the existence of a solution $(u,Y^{\phi},f)$ to the system of equations \eqref{30}, \eqref{38}, and \eqref{130}. Observe that this is in fact a coupled system, as the definition of $u$ depends on $\overline{g}$. In \cite{ChaKhuri}, a complete analysis of the equation \eqref{38} was performed assuming that $u$ is known a priori. It was shown that a smooth solution always exists (without blow-up) which possesses the desired asymptotics \eqref{40}-\eqref{42}. Moreover \eqref{30} turns out to be an inhomogeneous linear elliptic equation for $Y^{\phi}$ \cite{ChaKhuri} (if $u$ is independent of $Y^{\phi}$), and it was shown that the corresponding homogeneous equation has a unique bounded solution satisfying the asymptotic boundary condition \eqref{31}; the same techniques may be used to obtain the same result for the inhomogeneous equation here.

To end this section, we record the reduction statement for multiple black holes. We will denote the angular momentum and charges for each black hole by $\mathcal{J}_{i}$, $(q_{e})_{i}$, $(q_{b})_{i}$. Let
\begin{equation}\label{149}
\mathcal{F}(m,\mathcal{J}_{1},\ldots,\mathcal{J}_{N},(q_{e})_{1},\ldots,(q_{e})_{N},
(q_{b})_{1},\ldots,(q_{b})_{N})
\end{equation}
denote the numerical value of the action functional \eqref{121} evaluated at the harmonic map, from
$\mathbb{R}^{3}-\{\overline{\rho}=0\}$ to the complex two-dimensional hyperbolic space, which is expected to exist in analogy with the uncharged case \cite{ChruscielLiWeinstein}. Whether this value agrees with, or is bounded below by
\begin{equation}\label{150}
\sqrt{\frac{\mathcal{J}^{2}}{m^{2}}+q^{2}},
\end{equation}
where $q^{2}=q_{e}^{2}+q_{b}^{2}$ and
\begin{equation}\label{151}
\mathcal{J}=\sum_{n=1}^{N}\mathcal{J}_{n},\text{ }\text{ }\text{ }\text{ }\text{ }
q_{e}=\sum_{n=1}^{N}(q_{e})_{n},\text{ }\text{ }\text{ }\text{ }\text{ }
q_{b}=\sum_{n=1}^{N}(q_{b})_{n},
\end{equation}
is an important open problem. The proof of the following theorem is similar to that of Theorem \ref{thm2}.

\begin{theorem}\label{thm3}
Let $(M,g,k,E,B)$ be a smooth, simply connected, axially symmetric initial data set satisfying the charged dominant energy condition \eqref{2} and conditions \eqref{9}, and with $N+1$ ends, one designated asymptotically flat and the others either asymptotically flat or asymptotically cylindrical. If the system of equations \eqref{30}, \eqref{38}, \eqref{130} admits a smooth solution $(u,Y^{\phi},f)$ satisfying the asymptotics \eqref{31}, \eqref{40}-\eqref{42}, \eqref{131}-\eqref{133}, then
\begin{equation}\label{152}
m\geq\mathcal{F}(m,\mathcal{J}_{1},\ldots,\mathcal{J}_{N},(q_{e})_{1},\ldots,(q_{e})_{N},
(q_{b})_{1},\ldots,(q_{b})_{N}).
\end{equation}
\end{theorem}

\section{Further Applications}
\label{sec5} \setcounter{equation}{0}
\setcounter{section}{5}

Here the deformation of initial data for the Einstein-Maxwell system, given in previous sections, will be applied to two more problems. The first concerns the reduction argument associated with another geometric inequality for axisymmetric black holes, namely a lower bound for area in terms of mass, angular momentum, and charge, and the second involves improving upon the hypotheses required for the mass-angular momentum-charge inequality in the maximal case.

Assume that the initial data possess only two ends denoted $M_{end}^{\pm}$, such that $M_{end}^{+}$ is asymptotically flat and $M_{end}^{-}$ is either asymptotically flat or asymptotically cylindrical. Heuristic physical arguments \cite{DainKhuriWeinsteinYamada} relying on the cosmic censorship conjecture, lead to the following upper and lower bounds
\begin{equation}\label{153}
m^2-\frac{q^2}{2}-\sqrt{\left(m^2-\frac{q^2}{2}\right)^2-\frac{q^4}{4}-\mathcal{J}^2}\leq\frac{A_{min}}{8\pi}
\leq m^2-\frac{q^2}{2}+\sqrt{\left(m^2-\frac{q^2}{2}\right)^2-\frac{q^4}{4}-\mathcal{J}^2},
\end{equation}
where $A_{min}$ is the minimum area required to enclose $M_{end}^{-}$. In \cite{DainKhuriWeinsteinYamada} the lower bound is established in the maximal case, and it is also shown that equality occurs if and only if the initial data set is isometric to the $t=0$ slice of the extreme Kerr-Newman spacetime. The proof follows directly from the mass-angular momentum-charge inequality and the area-angular momentum-charge inequality \cite{ClementJaramilloReiris}. In the non-maximal case, the area-angular momentum-charge inequality has been established when $A_{min}$ is replaced by the area of a stable, axisymmetric, marginally outer trapped surface \cite{ClementJaramilloReiris}. Therefore, since we have shown how to reduce the non-maximal case of the mass-angular momentum-charge inequality to the problem of solving a coupled system of elliptic equations, an analogous lower bound for area may also be reduced to the same problem. That is, combining Theorem \ref{thm2} above with the proof of a Theorem 2.5 in \cite{DainKhuriWeinsteinYamada}, produces the following result.

\begin{theorem}\label{thm4}
Let $(M,g,k,E,B)$ be a smooth, simply connected, axially symmetric initial data set satisfying the charged dominant energy condition \eqref{2} and conditions \eqref{9}, and with two ends, one designated asymptotically flat and the other either asymptotically flat or asymptotically cylindrical. If the data possess a stable axisymmetric marginally outer trapped surface with area $\mathcal{A}$, and the system of equations \eqref{30}, \eqref{38}, \eqref{130} admits a smooth solution $(u,Y^{\phi},f)$ satisfying the asymptotics \eqref{31}, \eqref{40}-\eqref{42}, \eqref{131}-\eqref{133}, then
\begin{equation}\label{154}
\frac{\mathcal{A}}{8\pi}\geq m^2-\frac{q^2}{2}-\sqrt{\left(m^2-\frac{q^2}{2}\right)^2-\frac{q^4}{4}-\mathcal{J}^2},
\end{equation}
and equality holds if and only if $(M,g,k,E,B)$ arises from an embedding into the extreme Kerr-Newman spacetime.
\end{theorem}

Now let us return to the mass-angular momentum-charge inequality \eqref{15}. In the maximal case, the results of \cite{ChruscielCosta}, \cite{Costa}, \cite{SchoenZhou} assume that the matter density is nonnegative $\mu_{EM}\geq 0$, the current density vanishes $|J_{EM}|_{g}=0$, and that the 4-currents for the electric and magnetic fields (sources for the Maxwell equations) vanish; all of these conditions are satisfied in electrovacuum (the absence of matter fields other than the electromagnetic field). The later assumption concerning the 4-currents is imposed in order to secure the existence of potentials for the Maxwell field, and $|J_{EM}|_{g}=0$ is used to obtain a charged twist potential, as in Lemma \ref{lemma5}. The new observation that we will make here, is that these hypotheses can be weakened to $div_{g}E=div_{g}B=J_{EM}(\eta)=0$. This will be accomplished by deforming the electric and magnetic fields as in Section \ref{sec3}, while keeping the original metric and second fundamental form, that is $(M,g,k,E,B)\rightarrow (M,g,k,\overline{E},\overline{B})$. More precisely,  $(\overline{E},\overline{B})$ are given by the definitions in Section \ref{sec3} with $f\equiv 0$. In particular
\begin{equation}\label{155}
\overline{E}(e_{i})=E(e_{i}), \text{ }\text{ }\text{ }\text{ } \overline{B}(e_{i})=B(e_{i}),\text{ }\text{ }\text{ }\text{ }i=1,2,\text{ }\text{ }\text{ }\text{ }\overline{E}(e_{3})=\overline{B}(e_{3})=0,
\end{equation}
and according to Lemmas \ref{lemma1}-\ref{lemma4}
\begin{equation}\label{156}
\overline{q}_{e}=q_{e},\text{ }\text{ }\text{ }\text{ }\overline{q}_{b}=q_{b},\text{ }\text{ }\text{ }\text{ }\overline{J}_{EM}(\eta)=J_{EM}(\eta),\text{ }\text{ }\text{ }\text{ }div_{g}\overline{E}=div_{g}E,\text{ }\text{ }\text{ }\text{ }div_{g}\overline{B}=div_{g}B,
\end{equation}
\begin{equation}\label{157}
|\widetilde{E}|_{g}^{2}+|\widetilde{B}|_{g}^{2}
=|\overline{E}|_{g}^{2}+|\overline{B}|_{g}^{2},\text{ }\text{ }\text{ }\text{ }\text{ }
8\pi\left(\widetilde{\mu}_{EM}-\mu_{EM}\right)=E(e_{3})^{2}+B(e_{3})^{2}.
\end{equation}
Thus, if $div_{g}E=div_{g}B=J_{EM}(\eta)=0$ then Lemma \ref{lemma5} yields potentials such that
\begin{equation}\label{158}
d\chi=|\eta | \left(\be(e_{2})\theta^{1}
- \be(e_{1})\theta^{2}\right)=|\eta | \left(E(e_{2})\theta^{1}
- E(e_{1})\theta^{2}\right),
\end{equation}
\begin{equation}\label{159}
d\psi=|\eta | \left(\bm(e_{2})\theta^{1}
-  \bm(e_{1})\theta^{2}\right)=|\eta | \left(B(e_{2})\theta^{1}
-  B(e_{1})\theta^{2}\right),
\end{equation}
and
\begin{equation}\label{160}
d\omega=k(\eta)\times\eta-\chi d\psi
+\psi d\chi.
\end{equation}
With the aid of Lemma \ref{lemma6}, the scalar curvature may be suitably bounded from below
\begin{align}\label{161}
\begin{split}
R=&|k|^{2}+2(|\widetilde{E}|_{g}^{2}+|\widetilde{B}|_{g}^{2})+16\pi\widetilde{\mu}_{EM}\\
=&|k|^{2}+2(|\overline{E}|_{g}^{2}+|\overline{B}|_{g}^{2})+16\pi\mu_{EM}
+2\left(E(e_{3})^{2}+B(e_{3})^{2}\right)\\
\geq&2\frac{e^{6U-2\alpha}}{\rho^{4}}|\partial\omega
+\chi\partial\psi-\psi\partial\chi|^{2}
+2\frac{e^{4U-2\alpha}}{\rho^{2}}\left(|\partial\chi|^{2}
+|\partial\psi|^{2}\right)\\
&+16\pi\mu_{EM}
+2\left(E(e_{3})^{2}+B(e_{3})^{2}\right).
\end{split}
\end{align}
The inequality in the third line arises from the fact that unlike $\overline{k}$, $k$ may have more than just two nontrivial components. Therefore, the methods of \cite{ChruscielCosta}, \cite{Costa}, \cite{SchoenZhou} may be used to obtain the following result.

\begin{theorem}\label{thm5}
Let $(M,g,k,E,B)$ be a smooth, simply connected, axially symmetric, maximal initial data set satisfying $\mu_{EM}\geq 0$ and $J_{EM}(\eta)=0$, and with two ends, one designated asymptotically flat and the other either asymptotically flat or asymptotically cylindrical. Then
\begin{equation}\label{162}
m^2\geq\frac{q^{2}+\sqrt{q^{4}+4\mathcal{J}^{2}}}{2},
\end{equation}
and equality holds if and only if $(M,g,k,E,B)$ is isometric to the $t=0$ slice of the extreme Kerr-Newman spacetime.
\end{theorem}

\section{Appendix A: Alternate Proof of Lemma \ref{lemma4}}
\label{sec6} \setcounter{equation}{0}
\setcounter{section}{6}

Here we will verify the relation between the divergences of $(\widetilde{E},\widetilde{B})$ and $(\overline{E},\overline{B})$ by direct computation. The notation in the proof of Lemma \ref{lemma4}
will be utilized throughout this section.

Observe that
\begin{equation} \label{00020}
div_{g} \widetilde{E}=\nabla_{i}\widetilde{E}^{i} = X_{i}(\widetilde{E}^{i}) + g^{ij} \tg(\tilna_{X_{i}} X_{l}, X_{j} )\widetilde{E}^{l}.
\end{equation}
Each term of \eqref{00020} will now be computed separately. First note that with the help of \eqref{68} and \eqref{69},
\begin{align} \label{00021}
\begin{split}
X_{i}(\widetilde{E}^{i})
&= \overline{e}_{1} \left( \frac{\be^{1}+ u \overline{e}_{2}(f)\bm_{3}}{\sqrt{\volbarg}}   \right)
+ \overline{e}_{2} \left( \frac{\be^{2} - u \overline{e}_{1}(f)\bm_{3}}{\sqrt{\volbarg}}   \right)
\\
&= \frac{\overline{e}_{1}(\be^{1}+ u \overline{e}_{2}(f)\bm_{3})+\overline{e}_{2}(\be^{2} - u \overline{e}_{1}(f)\bm_{3})}{\sqrt{\volbarg}} +  {\sqrt{\volbarg}}\text{ }\widetilde{E}^{i}\overline{e}_{i}\left(\frac{1}{\sqrt{\volbarg}}\right).
\end{split}
\end{align}
Next consider the second term on the right-hand side of \eqref{00020} and compute
\begin{align} \label{00022}
\begin{split}
\tg(\tilna_{X_{i}} X_{l}, X_{j})
=& \tg\left(\tilna_{\overline{e}_{i}+ \overline{e}_{i}(f) \partial_{t}} (\overline{e}_{l}+ \overline{e}_{l}(f) \partial_{t}), \overline{e}_{j}+ \overline{e}_{j}(f) \partial_{t}\right) \\
=& \bg(\barna_{\overline{e}_{i}} \overline{e}_{l}, \overline{e}_{j})
	+ \overline{e}_{i}\overline{e}_{l}(f) \tg(\partial_{t}, \overline{e}_{j}+ \overline{e}_{j}(f) 			 \partial_{t} )
+ \overline{e}_{l}(f) \tg(\tilna_{\overline{e}_{i}} \partial_{t}, \overline{e}_{j}+ \overline{e}_{j}(f) 					 \partial_{t} )\\
  &  + \overline{e}_{j}(f) \tg(\tilna_{\overline{e}_{i}} \overline{e}_{l}, \partial_{t})
+ \overline{e}_{i}(f) \left(\tg(\tilna_{\partial_{t}} \overline{e}_{l}, \overline{e}_{j}+ \overline{e}_{j}(f)\partial_{t}) + \overline{e}_{l}(f) \tg(\tilna_{\partial_{t}} \partial_{t}, \overline{e}_{j}) \right),
\end{split}
\end{align}
where
\begin{equation} \label{00022.1}
\tg(\tilna_{\overline{e}_{i}} \partial_{t}, \overline{e}_{j})
= \tg\left(\tilna_{\overline{e}_{i}} (u n-\by), \overline{e}_{j}\right)
= u \bk(\overline{e}_{i}, \overline{e}_{j}) - \bg(\barna_{\overline{e}_{i}} \by, \overline{e}_{j}),
\end{equation}
\begin{equation} \label{00022.2}
\tg(\tilna_{\overline{e}_{i}} \overline{e}_{l}, \partial_{t})
= \tg(\tilna_{\overline{e}_{i}} \overline{e}_{l}, u n-\by)
=- u \bk(\overline{e}_{i}, \overline{e}_{l}) - \bg(\barna_{\overline{e}_{i}} \overline{e}_{l}, \by),
\end{equation}
\begin{equation} \label{00022.3}
\tg(\tilna_{\partial_{t}} \overline{e}_{l}, \overline{e}_{j})
= \tg(\tilna_{\overline{e}_{l}}\partial_{t} , \overline{e}_{j})
= u \bk(\overline{e}_{l}, \overline{e}_{j}) - \bg(\barna_{\overline{e}_{l}} \by, \overline{e}_{j}),
\end{equation}
\begin{equation} \label{00022.4}
\tg(\tilna_{ \partial_{t}} \overline{e}_{l} , \partial_{t})
= \tg(\tilna_{\overline{e}_{l} } \partial_{t} , \partial_{t})
= \frac{1}{2} \overline{e}_{l}(\tg_{tt}),
\end{equation}
and
\begin{equation} \label{00022.5}
\tg(\tilna_{ \partial_{t}} \partial_{t}, \overline{e}_{j} )
= -\tg(\tilna_{\overline{e}_{j} } \partial_{t} , \partial_{t})
= -\frac{1}{2} \overline{e}_{j}(\tg_{tt}).
\end{equation}
Substituting \eqref{00022.1}-\eqref{00022.5} into \eqref{00022} produces
\begin{align}\label{00022.6}
\begin{split}
\tg(\tilna_{X_{i}} X_{l}, X_{j})
=& \bg(\barna_{\overline{e}_{i}} \overline{e}_{l}, \overline{e}_{j})
	+ \overline{e}_{i}\overline{e}_{l}(f) \left( -\by(\overline{e}_{j}) + \overline{e}_{j}			 (f)\tg_{tt}\right) \\
&+ \overline{e}_{l}(f) \left( u \bk(\overline{e}_{i}, \overline{e}_{j})
	- \bg(\barna_{\overline{e}_{i}} \by , \overline{e}_{j})+ \frac{1}{2} \overline{e}_{j}(f)\overline{e}_{i}(\tg_{tt}) \right) \\
&+ \overline{e}_{j}(f) \left(-u \bk(\overline{e}_{i}, \overline{e}_{l})
	- \bg(\barna_{\overline{e}_{i}} \overline{e}_{l} , \by)  \right)  \\
&+ \overline{e}_{i}(f) \left(u \bk(\overline{e}_{j}, \overline{e}_{l})
	- \bg(\barna_{\overline{e}_{l}} \by ,\overline{e}_{j})
	+ \frac{1}{2} \overline{e}_{j}(f)\overline{e}_{l}(\tg_{tt})
	- \frac{1}{2}\overline{e}_{l}(f)\overline{e}_{j}(\tg_{tt}) \right).
\end{split}
\end{align}

The following simple identities will aid in taking the trace of \eqref{00022.6}, namely
\begin{equation}\label{00022.6.1}
u tr_{\bg} \bk = \delta^{ij}\bg(\barna_{\overline{e}_{i}} \by, \overline{e}_{j}) = 0, \text{ }\text{ }\text{ }\text{ }\text{ }u\bk(w, w)= \bg(\barna_{w} \by, w).
\end{equation}
We then find that
\begin{align}\label{00022.7}
\begin{split}
 g^{ij}\tg(\tilna_{X_{i}} X_{l}, X_{j})
=& \left(\delta^{ij} - \frac{|\by|_{\bg}^{2}\delta^{3i}\delta^{3j}}{u^{2}} + w(\overline{\theta}^{i})w(\overline{\theta}^{j}) \right)\tg(\tilna_{X_{i}} X_{l}, X_{j}) \\
=& \bg(\barna_{\overline{e}_{i}} \overline{e}_{l}, \overline{e}_{i})
   - \frac{1}{u^{2}} \bg(\barna_{\by} \overline{e}_{l}, \by)
   + \bg(\barna_{w} \overline{e}_{l}, w) \\
&+\tg_{tt}\left[ \barna f\left( \overline{e}_{l}(f) \right)
	+ w \left( \overline{e}_{l}(f) \right) w(f) \right] - \bg(\by,w)w \left( \overline{e}_{l}(f) \right)  \\
&+ \left(- \bg(\barna_{\barna f} \overline{e}_{l} , \by)
							- \bg(\barna_{\overline{e}_{l}} \by , \barna f)
	+ \frac{1}{2}|\barna f|_{\bg}^{2} \overline{e}_{l}(\tg_{tt})\right) \\
&+  w(f) \left(- \bg(\barna_{w} \overline{e}_{l} , \by)
							- \bg(\barna_{\overline{e}_{l}} \by ,w)
	+ \frac{1}{2} w(f)\overline{e}_{l}(\tg_{tt})\right).
\end{split}
\end{align}
Applying the expression
\begin{equation}\label{00022.7.1}
 w = \frac{u^{2} \barna f + \by}{u \sqrt{\volbarg}}
\end{equation}
yields
\begin{align} \label{00022.8}
\begin{split}
&\tg_{tt}\left[ \barna f\left( \overline{e}_{l}(f) \right)
	+ w \left( \overline{e}_{l}(f) \right) w(f) \right] - \bg(\by,w)w \left( \overline{e}_{l}(f) \right)\\
=& \barna f(\overline{e}_{l}(f)) \left(\widetilde{g}_{tt} + \frac{u^{2}|\barna f|_{\bg}^{2}\tg_{tt}}{\volbarg}
- \frac{|\by|_{\bg}^{2}}{\volbarg} \right)\\
=& -\frac{u^{2}\barna f(\overline{e}_{l}(f))}{\volbarg},
\end{split}
\end{align}
\begin{equation} \label{00022.9}
|\barna f|_{\bg}^{2} \overline{e}_{l}(\tg_{tt})+ w(f)^{2}\overline{e}_{l}(\tg_{tt})
=\frac{|\barna f|_{\bg}^{2}\overline{e}_{l}(\tg_{tt})}{\volbarg }  ,
\end{equation}
\begin{align} \label{00022.10}
\begin{split}
&- \frac{1}{u^{2}} \bg(\barna_{\by} \overline{e}_{l}, \by) + \bg(\barna_{w} \overline{e}_{l}, w)
- \bg(\barna_{\barna f} \overline{e}_{l} , \by)
-w(f)\bg(\barna_{w} \overline{e}_{l} , \by) \\
=& \bg\left(\barna_{w} \overline{e}_{l}, \frac{u \barna f}{\sqrt{\volbarg}}\right),
\end{split}
\end{align}
and
\begin{align} \label{00022.11}
\begin{split}
 \bg(\barna_{\overline{e}_{l}} \by , \barna f)+w(f) \bg(\barna_{\overline{e}_{l}} \by ,w)
=&  \frac{\bg(\barna_{\overline{e}_{l}} \by , \barna f)}{\volbarg}
+ \frac{|\barna f|_{\bg}^{2}\bg(\barna_{\overline{e}_{l}} \by, \by)}{\volbarg}  \\
=& \frac{\bg(\barna_{ \by} \overline{e}_{l}, \barna f)}{\volbarg}
+ \frac{|\barna f|_{\bg}^{2}\overline{e}_{l}(|\by|_{\bg}^{2})}{2\left(\volbarg\right)} ,
\end{split}
\end{align}
where we have used $\by=\byp \eta$ and $\bg(\barna f, \eta)=0$ in the last line of \eqref{00022.11}.
By substituting \eqref{00022.8}-\eqref{00022.11} into \eqref{00022.7} we obtain
\begin{align} \label{00023}
\begin{split}
g^{ij}\tg(\tilna_{X_{i}} X_{l}, X_{j})
&=  \bg(\barna_{\overline{e}_{i}} \overline{e}_{l}, \overline{e}_{i})
  -\frac{u|\barna f|_{\bg}^{2} \overline{e}_{l}(u)}{\volbarg }
  -\frac{u^{2}\barna f(\overline{e}_{l}(f))}{\volbarg}
+ \frac{u^{2}}{\volbarg}\bg(\barna_{\barna f} \overline{e}_{l}, \barna f) \\
&= \bg(\barna_{\overline{e}_{i}} \overline{e}_{l}, \overline{e}_{i})
  -\frac{u|\barna f|_{\bg}^{2} \overline{e}_{l}(u)}{\volbarg }
  -\frac{u^{2}}{\volbarg} (Hess_{\bg}f)(\barna f, \overline{e}_{l}) \\
&=  \bg(\barna_{\overline{e}_{i}} \overline{e}_{l}, \overline{e}_{i})
   - \overline{e}_{l}\left(\frac{1}{\sqrt{\volbarg}}\right)\sqrt{\volbarg}.
\end{split}
\end{align}
In particular, when $l=3$,  \eqref{00023} becomes
\begin{equation} \label{00024}
 g^{ij} \tg(\tilna_{X_{i}} X_{3}, X_{j} )
= \bg(\barna_{\overline{e}_{i}} \overline{e}_{3}, \overline{e}_{i})
=\frac{1}{|\eta|} \bg(\barna_{\overline{e}_{i}} \eta, \overline{e}_{i})=0.
\end{equation}

Now employ \eqref{00021}, \eqref{00023}, \eqref{00024} to evaluate \eqref{00020} as follows,
\begin{align} \label{00025}
\begin{split}
div_{g} \widetilde{E}
=& X_{i}(\widetilde{E}^{i}) + g^{ij} \tg(\tilna_{X_{i}} X_{1}, X_{j} )\widetilde{E}^{1} + g^{ij} \tg(\tilna_{X_{i}} X_{2}, X_{j} )\widetilde{E}^{2}  \\
=& \frac{\overline{e}_{1}(\be^{1}+ u \overline{e}_{2}(f)\bm_{3})+\overline{e}_{2}(\be^{2} - u \overline{e}_{1}(f)\bm_{3})}{\sqrt{\volbarg}}  \\
&+ \frac{\be^{1}+ u \overline{e}_{2}(f)\bm_{3}}{\sqrt{\volbarg}} \bg(\barna_{\overline{e}_{i}} \overline{e}_{1}, \overline{e}_{i})
+\frac{\be^{2}- u \overline{e}_{2}(f)\bm_{3}}{\sqrt{\volbarg}} \bg(\barna_{\overline{e}_{i}} \overline{e}_{2}, \overline{e}_{i}).
\end{split}
\end{align}
Similarly
\begin{align} \label{00026}
\begin{split}
div_{g} \widetilde{B}
=& \frac{\overline{e}_{1}(\bm^{1}- u \overline{e}_{2}(f)\be_{3})+\overline{e}_{2}(\bm^{2} + u \overline{e}_{1}(f)\be_{3})}{\sqrt{\volbarg}}  \\
&+ \frac{\bm^{1}- u \overline{e}_{2}(f)\be_{3}}{\sqrt{\volbarg}} \bg(\barna_{\overline{e}_{i}} \overline{e}_{1}, \overline{e}_{i})
+\frac{\bm^{2}+u \overline{e}_{2}(f)\be_{3}}{\sqrt{\volbarg}} \bg(\barna_{\overline{e}_{i}} \overline{e}_{2}, \overline{e}_{i}).
\end{split}
\end{align}
Therefore since $\be_{3}=\bm_{3}=0$, we have
\begin{equation} \label{00027}
div_{g} \widetilde{E} = \frac{div_{\bg}\be}{\sqrt{\volbarg}}, \text{ }\text{ }\text{ }\text{ }\text{ }
div_{g} \widetilde{B} = \frac{div_{\bg} \bm}{\sqrt{\volbarg}}.
\end{equation}
The desired result now follows from \eqref{79}.

\section{Appendix B: Two Versions of the Charged Twist Potential}
\label{sec7} \setcounter{equation}{0}
\setcounter{section}{7}

In Lemma \ref{lemma5}, under the assumption $\overline{J}_{EM}(\eta)=0$, the existence of a charged twist potential was obtained with the aid of
the potentials $(\overline{\chi},\overline{\psi})$ for the electromagnetic field, namely
\begin{equation}\label{b.1}
d\overline{\omega}=\overline{k}(\eta)\times\eta-\overline{\chi}d\overline{\psi}
+\overline{\psi}d\overline{\chi}.
\end{equation}
Here we will show that the charged twist potential may be constructed from the vector potential instead. That is
\begin{equation}\label{b.2}
d\overline{\omega}=\overline{k}(\eta)\times\eta+2[\overline{A}\times(\eta\times\overline{E})]\times\eta,
\end{equation}
where $\overline{B}=\overline{\nabla}\times\overline{A}$. Note that due to topological considerations, one must remove a Dirac string in order to construct $\overline{A}$, as described in Section \ref{sec1}.

We begin by deriving the relationship between $\overline{\psi}$ and $\overline{A}$. Observe that by \eqref{109},
\begin{equation}\label{b.3}
\overline{e}_{1}(\overline{\psi})=|\eta|\overline{B}(\overline{e}_{2}),\text{ }\text{ }\text{ }\text{ }\text{ }
\overline{e}_{2}(\overline{\psi})=-|\eta|\overline{B}(\overline{e}_{1}).
\end{equation}
Moreover if
\begin{equation}\label{b.4}
\overline{A}=\overline{A}(\overline{e}_{1})\overline{\theta}^{1}+\overline{A}(\overline{e}_{2})\overline{\theta}^{2}
+\overline{A}(\overline{e}_{3})\overline{\theta}^{3},
\end{equation}
then \eqref{113}, \eqref{114} imply that
\begin{align}\label{b.5}
\begin{split}
\overline{B}
=\ast d\overline{A}=&\left[\overline{e}_{2}(\overline{A}(\overline{e}_{3}))+
\overline{e}_{2}(\log|\eta|)\overline{A}(\overline{e}_{3})\right]
\overline{\theta}^{1}
-\left[\overline{e}_{1}(\overline{A}(\overline{e}_{3}))+
\overline{e}_{1}(\log|\eta|)\overline{A}(\overline{e}_{3})\right]
\overline{\theta}^{2}\\
&+\left[\overline{e}_{1}(\overline{A}(\overline{e}_{2}))+
\overline{e}_{1}(\log e^{-\overline{U}+\overline{\alpha}})\overline{A}(\overline{e}_{2})
-\overline{e}_{2}(\overline{A}(\overline{e}_{1}))-
\overline{e}_{2}(\log e^{-\overline{U}+\overline{\alpha}})\overline{A}(\overline{e}_{1})\right]
\overline{\theta}^{3}\\
&-|\eta|e^{2\overline{U}-2\overline{\alpha}}(A_{\overline{\rho},\overline{z}}
-A_{\overline{z},\overline{\rho}})\overline{A}(\overline{e}_{3})\overline{\theta}^{3}.
\end{split}
\end{align}
Combining this with \eqref{b.3} produces
\begin{equation}\label{b.6}
\overline{e}_{1}(\overline{\psi})
=-|\eta|\left[\overline{e}_{1}(\overline{A}(\overline{e}_{3}))+
\overline{e}_{1}(\log|\eta|)\overline{A}(\overline{e}_{3})\right],\text{ }\text{ }\text{ }\text{ }
\overline{e}_{2}(\overline{\psi})=-|\eta|\left[\overline{e}_{2}(\overline{A}(\overline{e}_{3}))+
\overline{e}_{2}(\log|\eta|)\overline{A}(\overline{e}_{3})\right].
\end{equation}
Hence
\begin{equation}\label{b.7}
\overline{\psi}=-|\eta|\overline{A}(\overline{e}_{3})+C
\end{equation}
for some constant $C$, so that the third component of $\overline{A}$ is determined by $\overline{\psi}$. One more component of $\overline{A}$ may be determined from the equation $\overline{B}(\overline{e}_{3})=0$ and \eqref{b.5}, and
the remaining component of $\overline{A}$ remains undetermined due to gauge invariance.

In order to show that the potentials defined by \eqref{b.1} and \eqref{b.2} are equivalent, we must show that the right-hand sides of these equations differ by an exact 1-form. In fact, this conclusion arises from \eqref{109} and \eqref{b.7} as follows
\begin{align}\label{b.8}
\begin{split}
-\overline{\chi}d\overline{\psi}+\overline{\psi}d\overline{\chi}+d(\overline{\chi}\overline{\psi}-2C\overline{\chi})
&=2\overline{\psi}d\overline{\chi}-2Cd\overline{\chi}\\
&=-2|\eta|\overline{A}(\overline{e}_{3})d\overline{\chi}\\
&=-2|\eta|^{2}\overline{A}(\overline{e}_{3})\overline{E}(\overline{e}_{2})\overline{\theta}^{1}
+2|\eta|^{2}\overline{A}(\overline{e}_{3})\overline{E}(\overline{e}_{1})\overline{\theta}^{2}\\
&=2\overline{\epsilon}_{ijl}\overline{\epsilon}_{cba}\overline{E}^{c}
\overline{\epsilon}^{jpb}\overline{A}_{p}\eta^{l}\eta^{a}\overline{\theta}^{i}\\
&=2[\overline{A}\times(\eta\times\overline{E})]\times\eta.
\end{split}
\end{align}

We now record how the charged twist potential \eqref{10} encodes angular momentum \eqref{13}. Consider the polar coordinate form $(r,\theta,\phi)$ of Brill coordinates, where $\rho=r\sin\theta$ and $z=r\cos\theta$. The metric may then be expressed by
\begin{equation}\label{b1}
g=e^{-2U+2\alpha}(dr^{2}+r^{2}d\theta^{2})+e^{-2U}r^{2}\sin^{2}\theta(d\phi+A_{r}dr+A_{\theta}d\theta)^{2},
\end{equation}
and an orthonormal basis is given by
\begin{equation}\label{b2}
e_{r}=e^{U-\alpha}(\partial_{r}-A_{r}\partial_{\phi}),\text{ }\text{ }\text{ }\text{ }e_{\theta}=\frac{e^{U-\alpha}}{r}(\partial_{\theta}-A_{\theta}\partial_{\phi}),
\text{ }\text{ }\text{ }\text{ }e_{\phi}=\frac{e^{U}}{r\sin\theta}\partial_{\phi}.
\end{equation}
From \eqref{10} it follows that
\begin{align}\label{b3}
\begin{split}
\frac{e^{U-\alpha}}{r}\partial_{\theta}\omega &=e_{\theta}(\omega)\\
&=-|\eta|^{2}\epsilon(e_{r},e_{\theta},e_{\phi})\left[k(e_{r},e_{\phi})
+2\epsilon(e_{r},e_{\theta},e_{\phi})E(e_{r})\epsilon(e_{r},e_{\phi},e_{\theta})\vec{A}(e_{\phi})\right]\\
&=-e^{-U}r\sin\theta\left[k(e_{r},\partial_{\phi})-2E(e_{r})\vec{A}(\partial_{\phi})\right],
\end{split}
\end{align}
or rather
\begin{equation}\label{b4}
k(e_{r},\partial_{\phi})-2E(e_{r})\vec{A}(\partial_{\phi})=-\frac{e^{2U-\alpha}}{r^{2}\sin\theta}\partial_{\theta}\omega.
\end{equation}
Hence, if there are only two ends, we may apply Lemma 2.1 in \cite{DainKhuriWeinsteinYamada} to obtain
\begin{align}\label{b5}
\begin{split}
\mathcal{J}&=\frac{1}{8\pi}\int_{S_{\infty}}(k_{ij}-(Tr k)g_{ij})\nu^{i}\eta^{j}
-\frac{1}{4\pi}\int_{S_{\infty}}(E_{i}\nu^{i})(\vec{A}_{j}\eta^{j})\\
&=\lim_{r\rightarrow0}\frac{1}{8\pi}\int_{\partial B(r)}\left[k(\partial_{\phi},e_{r})
-2E(e_{r})\vec{A}(\partial_{\phi})\right]dA\\
&=\lim_{r\rightarrow0}\frac{1}{8\pi}\int_{\partial B(1)}\left[k(\partial_{\phi},e_{r})
-2E(e_{r})\vec{A}(\partial_{\phi})\right]e^{-2U+\alpha}r^{2}\sin\theta d\theta d\phi\\
&=-\lim_{r\rightarrow0}\frac{1}{8\pi}\int_{\partial B(1)}\partial_{\theta}\omega d\theta d\phi\\
&=\frac{1}{4}(\omega|_{I_{+}}-\omega|_{I_{-}}).
\end{split}
\end{align}
Observe that the expression for angular momentum here appears to differ from that of \eqref{7}, in that
an extra contribution from the electromagnetic field, representing the total angular momentum of the electromagnetic field, is present. It turns out that a particular cancelation \cite{DainKhuriWeinsteinYamada} forces this term to vanish, and hence \eqref{7} and \eqref{b5} yield the same value.

\section{Appendix C: Comparing Asymptotics with the Kerr-Newman Example}
\label{sec8} \setcounter{equation}{0}
\setcounter{section}{8}

In this section the prescribed asymptotics for $Y^{\phi}$ and $(E,B)$ will be compared with the corresponding
quatities in the Kerr-Newman spacetime. Recall that in Boyer-Lindquist coordinates the Kerr-Newmann metric takes the form
\begin{equation} \label{c1}
 -\frac{\Delta - a^{2}\sin^{2}\theta}{\Sigma} dt^{2}
 + \frac{2a\sin^{2}\theta}{\Sigma}\left( \widetilde{r}^{2}+a^{2}-\Delta \right)dtd\phi
+  \frac{(\widetilde{r}^{2}+a^{2})^{2} - \Delta a^{2}\sin^{2}\theta}{\Sigma} \sin^{2}\theta d\phi^{2}
+ \frac{\Sigma}{\Delta} d\widetilde{r}^{2} + \Sigma d\theta^{2}
\end{equation}
where
\begin{equation}\label{c2}
\Delta = \widetilde{r}^{2} + a^{2} +q^{2} -2m\widetilde{r},\text{ }\text{ }\text{ }\text{ }\text{ } \Sigma = \widetilde{r}^{2} + a^{2}\cos^{2}\theta,
\end{equation}
and the electromagnetic 4-potential is given by
\begin{equation} \label{c3}
\mathbf{A} = -\frac{q_{e}\widetilde{r}}{\Sigma} \left( dt+ a\sin^{2}\theta d\phi \right)
- \frac{q_{b}\cos\theta}{\Sigma} \left( a dt + (\widetilde{r}^{2}+a^{2}) d \phi \right),
\end{equation}
The event horizon is located at the larger of the two solutions to the quadratic equation
$\Delta=0$, namely $\widetilde{r}_{+}=m+\sqrt{m^{2}-a^{2}-q^{2}}$, where the angular momentum is given by $\mathcal{J} = ma$. For $\widetilde{r}>\widetilde{r}_{+}$ it holds that $\Delta>0$, so that a new radial coordinate may be defined by
\begin{equation}\label{c4}
r=\frac{1}{2}(\widetilde{r}-m+\sqrt{\Delta}),
\end{equation}
or rather
\begin{align} \label{c5}
\begin{split}
\widetilde{r} &= r + m + \frac{m^{2}-a^{2}-q^{2}}{4r},\text{ }\text{ }\text{ }\text{ }\text{ } m^{2} \neq a^{2}+q^{2}  \\
\widetilde{r} &= r + m, \text{ }\text{ }\text{ }\text{ }\text{ } m^{2}=a^{2}+q^{2}.
\end{split}
\end{align}
Note that the new coordinate is defined for $r>0$, and a critical point for the right-hand side of \eqref{c5} ($m^{2}\neq a^{2} + q^{2}$) occurs at the horizon, so that two isometric copies of the outer region are encoded on this interval. The coordinates $(r,\theta,\phi)$ then form a Brill coordinate system.

Observe that
\begin{equation} \label{c6}
\byp= g^{\phi \phi} Y_{\phi} = -\frac{a( \widetilde{r}^{2}+a^{2}-\Delta)}{(\widetilde{r}^{2}+a^{2})^{2} -\Delta a^{2}\sin^{2}\theta}.
\end{equation}
Therefore at spatial infinity
\begin{equation} \label{318}
\byp \sim -\frac{2ma}{r^{3}} \quad \textrm{as} \quad r \to \infty,
\end{equation}
which is consistent with \eqref{31}.
Furthermore,
\begin{equation} \label{c7}
\begin{split}
\byp &= O(r^{3}),\text{ }\text{ }\text{ }\text{ }  m^{2} \neq a^{2}+q^{2}, \text{ }\text{ }\text{ as }\text{ }\text{ }r\rightarrow 0,    \\
\byp &= -\frac{a(2m^{2}-q^{2})}{(m^{2}+a^{2})^{2}}+O(r), \text{ }\text{ }\text{ }\text{ }  m^{2} = a^{2} + q^{2}, \text{ }\text{ }\text{ as }\text{ }\text{ }r\rightarrow 0.
\end{split}
\end{equation}
Note that the solution of \eqref{30} is expected to asymptotically converge to a constant as $r\rightarrow 0$ \cite{ChaKhuri}. This is consistent with the second equation in \eqref{c7}, but not the first. The reason for the inconsistency is that the lapse function for the Kerr-Newman spacetime does not satisfy the required asymptotics \eqref{132}, whereas the lapse function
for the extreme Kerr-Newman spacetime does satisfy the desired asymptotics \eqref{133}; the lapse function is
\begin{align} \label{c8}
\begin{split}
u^{2} =& -\widetilde{g}_{tt} + |Y|^{2} \\
=& \frac{\Delta -a^{2}\sin^{2}\theta}{\Sigma}
+ \frac{a^{2}\sin^{2}\theta (\widetilde{r}^{2}+a^{2}-\Delta)^{2}}{\Sigma \left( (\widetilde{r}^{2}+a^{2})^{2} - \Delta a^{2} \sin^{2} \theta \right)} \\
=& \frac{\Delta \Sigma}{(\widetilde{r}^{2}+a^{2})^{2} - \Delta a^{2} \sin^{2} \theta}.
\end{split}
\end{align}

We now compute the induced electric and magnetic field on the $t=0$ slice.
The field strength is given by $F_{ij} = \partial_{i} \mathbf{A}_{j} -\partial_{j} \mathbf{A}_{i}$. Since $\partial_{t}$ and $\partial_{\phi}$ are Killing fields, it is straightforward to check that $F_{\widetilde{r} \theta} = F_{\phi t} =0$, which is equivalent to $E(\partial_{\phi})= B(\partial_{\phi})=0$. Next observe that
\begin{equation} \label{c9}
F_{\widetilde{r} t} = \partial_{\widetilde{r}} \mathbf{A}_{t}
= \frac{q_{e}}{\Sigma} \left( -1 + \frac{2 \widetilde{r}^{2}}{\Sigma}\right)
 + \frac{2 q_{b} \widetilde{r}a \cos \theta}{\Sigma^{2}},
\end{equation}
\begin{equation} \label{c10}
F_{\theta t} =\partial_{\theta}\mathbf{A}_{t}
= -\frac{a \sin\theta}{\Sigma} \left( \frac{2 q_{e} \widetilde{r} a \cos \theta}{\Sigma} + q_{b} \left( -1 + \frac{2 \widetilde{r}^{2}}{\Sigma}\right) \right),
\end{equation}
\begin{equation} \label{c11}
F_{\widetilde{r} \phi} = \partial_{\widetilde{r}} \mathbf{A}_{\phi}
= \frac{a \sin^{2}\theta}{\Sigma} \left( q_{e} \left( -1 + \frac{2 \widetilde{r}^{2}}{\Sigma}\right) + \frac{2 q_{b} \widetilde{r} a \cos \theta}{\Sigma} \right),
\end{equation}
\begin{equation} \label{c12}
F_{\theta \phi} = \partial_{\theta} \mathbf{A}_{\phi}
= \frac{(\widetilde{r}^{2}+a^{2})\sin\theta}{\Sigma} \left(-\frac{2 q_{e} \widetilde{r} a \cos \theta}{\Sigma} + q_{b} \left( -1 + \frac{2 \widetilde{r}^{2}}{\Sigma}\right) \right).
\end{equation}
Since $E(\cdot)= F(\cdot, n) = F(\cdot, u^{-1}(\partial_{t}+ \overline{Y}))$ we have
\begin{align} \label{c13}
\begin{split}
 E_{\widetilde{r}}&= \frac{1}{u}\overline{F}_{\widetilde{r} t} + \frac{\byp}{u} \overline{F}_{\widetilde{r} \phi} \\
&= \frac{1+ a \sin^{2} \theta \byp}{u} \left( \frac{q_{e}(\widetilde{r}^{2}-a^{2}\cos^{2}\theta)}{\Sigma^{2}}
+ \frac{2q_{b} \widetilde{r}a\cos \theta }{\Sigma^{2}}\right) \\
&=\frac{(\widetilde{r}^{2}+a^{2})\sin\theta}{\sqrt{\Delta}|\partial_{\phi}|} \left( \frac{q_{e}(\widetilde{r}^{2}-a^{2}\cos^{2}\theta)}{\Sigma^{2}} + \frac{2q_{b}\widetilde{r}a\cos \theta }{\Sigma^{2}}\right),
\end{split}
\end{align}
and
\begin{align} \label{c14}
\begin{split}
E_{\theta} &=\frac{1}{u}\overline{F}_{\theta t} + \frac{\byp}{u} \overline{F}_{\theta \phi} \\
&= \frac{\sqrt{\Delta} a \sin^{2}\theta}{|\partial_{\phi}|} \left( -\frac{2q_{e}\widetilde{r}a\cos \theta}{\Sigma^{2}} + \frac{q_{b}(\widetilde{r}^{2}-a^{2}\cos^{2}\theta) }{\Sigma^{2}}\right).
\end{split}
\end{align}
Also since $B(|\partial_{\widetilde{r}}|^{-1}\partial_{\widetilde{r}})= F(|\partial_{\theta}|^{-1}\partial_{\theta},|\partial_{\phi}|^{-1}\partial_{\phi})$ and $B(|\partial_{\theta}|^{-1}\partial_{\theta})= -F(|\partial_{\widetilde{r}}|^{-1}\partial_{\widetilde{r}},
|\partial_{\phi}|^{-1}\partial_{\phi})$, applying \eqref{c11} and \eqref{c12} yields
\begin{equation} \label{c15}
B_{\widetilde{r}} = \frac{(\widetilde{r}^{2}+a^{2})\sin\theta}{\sqrt{\Delta}|\partial_{\phi}|} \left( -\frac{2q_{e}\widetilde{r}a\cos \theta }{\Sigma^{2}} + \frac{q_{b}(\widetilde{r}^{2}-a^{2}\cos^{2}\theta)}{\Sigma^{2}}\right),
\end{equation}
\begin{equation} \label{c16}
B_{\theta} = -\frac{\sqrt{\Delta} a \sin^{2}\theta}{|\partial_{\phi}|} \left( \frac{q_{e}(\widetilde{r}^{2}-a^{2}\cos^{2}\theta)}{\Sigma^{2}} + \frac{2 q_{b}\widetilde{r}a\cos \theta }{\Sigma^{2}}\right).
\end{equation}

First consider the case when $m^{2} \neq q^{2} + a^{2}$, in which there are two asymptotically flat ends. In Brill coordinates
\begin{equation}\label{c17}
E_{r} = \frac{\partial \widetilde{r}}{\partial r} E_{\widetilde{r}}
= \frac{(\widetilde{r}^{2}+a^{2})\sin\theta}{\sqrt{\Delta}|\partial_{\phi}|}
\left( 1 - \frac{m^{2}-a^{2}-q^{2}}{4r^{2}} \right)
 \left( \frac{q_{e}(\widetilde{r}^{2}-a^{2}\cos^{2}\theta)}{\Sigma^{2}} + \frac{2q_{b}\widetilde{r}a\cos \theta }{\Sigma^{2}}\right),
\end{equation}
and similarly
\begin{equation}\label{c18}
B_{r} = \frac{(\widetilde{r}^{2}+a^{2})\sin\theta}{\sqrt{\Delta}|\partial_{\phi}|}
\left( 1 - \frac{m^{2}-a^{2}-q^{2}}{4r^{2}} \right)
 \left( -\frac{2q_{e}\widetilde{r}a\cos \theta }{\Sigma^{2}} + \frac{q_{b}(\widetilde{r}^{2}-a^{2}\cos^{2}\theta)}{\Sigma^{2}}\right).
\end{equation}
Therefore
\begin{equation}\label{c19}
E_{r} \sim \frac{q_{e}}{r^{2}},\text{ }\text{ }\text{ }\text{ }
E_{\theta} \sim \frac{q_{b}a \sin \theta }{r^{2}},  \text{ }\text{ }\text{ }\text{ }
B_{r} \sim \frac{q_{b}}{r^{2}},\text{ }\text{ }\text{ }\text{ }
B_{\theta} \sim  -\frac{q_{e}a \sin \theta}{r^{2}},\text{ }\text{ }\text{ as }\text{ }\text{ }
r\rightarrow\infty,
\end{equation}
and since $\widetilde{r} = O(r^{-1})$ for small $r$
\begin{equation} \label{c20}
E_{r} = O(1),\text{ }\text{ }\text{ }\text{ }
E_{\theta} = O(r^{2}),  \text{ }\text{ }\text{ }\text{ }
B_{r} = O(1),\text{ }\text{ }\text{ }\text{ }
B_{\theta} = O(r^{2}),\text{ }\text{ }\text{ as }\text{ }\text{ }
r\rightarrow 0.
\end{equation}
This is consistent with \eqref{4} and \eqref{18.1}. Now let us consider the extreme case  $m^{2}= a^{2}+ q^{2}$, in which there is one asymptotically flat end and one asymptotically cylindrical end. Here the radial coordinate change is merely a translation $\widetilde{r}= r+m$, so that
\begin{equation} \label{c21}
E_{r} = O(r^{-1}),\text{ }\text{ }\text{ }\text{ }
E_{\theta} = O(\rho)  ,  \text{ }\text{ }\text{ }\text{ }
B_{r} = O(r^{-1}),\text{ }\text{ }\text{ }\text{ }
B_{\theta} = O(\rho),\text{ }\text{ }\text{ as }\text{ }\text{ }
r\rightarrow 0,
\end{equation}
where $\rho = r \sin \theta$. This is consistent with \eqref{18.2}.

\end{document}